\documentclass[a4paper,onecolumn,11pt,unpublished]{quantumarticle}
\pdfoutput=1
\usepackage[utf8]{inputenc}
\usepackage[english]{babel}
\usepackage[T1]{fontenc}
\usepackage{amsmath}
\usepackage{amssymb}
\usepackage{physics}

\usepackage{hyperref}
\usepackage[table]{xcolor}
\usepackage{colortbl}
\usepackage{booktabs}
\usepackage{csquotes}
\usepackage{enumerate}
\usepackage{enumitem}
\usepackage{amsthm}
\usepackage{caption}
\usepackage{subcaption}

\usepackage{selinput}
\SelectInputMappings{
  adieresis={ä}, % etc.
  acute={´},
}
\usepackage{amsmath}  % if not already loaded

\usepackage[english]{babel}
\usepackage[T1]{fontenc}
\usepackage{lipsum}
\usepackage{mdframed}
\usepackage{booktabs}
\usepackage{adjustbox}
\usepackage{longtable}

\usepackage[
  backend=biber,
  style=numeric,
  sorting=none      % or 'nyt' if you want author–year sorting
]{biblatex}
\usepackage{xcolor}
\usepackage[table]{xcolor}
\usepackage{booktabs}
\usepackage{listings}
\usepackage{siunitx}
\usepackage{multicol}
\lstdefinestyle{TinyPython}{
  language        = Python,
  basicstyle      = \ttfamily\fontsize{5pt}{6pt}\selectfont,  % 5 pt type
  numbers         = left,
  numberstyle     = \ttfamily\fontsize{5pt}{6pt}\selectfont,
  stepnumber      = 1,
  numbersep       = 3pt,
  breaklines      = true,
  breakatwhitespace = true,
  showstringspaces = false,
  frame           = single,
  xleftmargin     = 0pt,
  xrightmargin    = 0pt,
  aboveskip       = 0pt,
  belowskip       = 0pt,
  tabsize         = 2,
  columns         = flexible,
  keepspaces      = true,
}
\usepackage{pifont}  
\usepackage{newunicodechar}
\newunicodechar{✖}{\ding{55}}
\usepackage[margin=1in]{geometry}   % Adjust page margins
\usepackage{setspace}              % For line spacing
\usepackage{amssymb} 
\usepackage{bm}
\usepackage{amsmath}
\usepackage{amsthm}  % Math symbols              % Theorem and proof environments
\usepackage{tikz}            % TikZ is required by quanti
\usepackage{quantikz} 
\usepackage{enumitem}% For drawing quantum circuits
\usepackage{setspace}              % For line spacing
\usepackage{bm}
\usepackage{amsmath,amssymb,amsthm}             % Theorem and proof environments
\usepackage{tikz}            % TikZ is required by quanti
\usepackage{quantikz} 
\usepackage{enumitem}% For drawing quantum circuits
\usepackage{graphicx} 
\usepackage{algpseudocode}
\usepackage{algorithm}
\usepackage{algorithmicx}
\usepackage{amsthm}
\usepackage{xcolor}
\usepackage[most]{tcolorbox}
\usepackage{pgfplots}
\pgfplotsset{compat=1.17}

\usepackage{booktabs}
\usepackage{float}                 % For [H] float placement
\usepackage{hyperref}              % Clickable links/references
\hypersetup{colorlinks=true,linkcolor=blue!50!black,citecolor=blue!50!black,urlcolor=blue!55!black,bookmarksopen=true,bookmarksnumbered=true,pdfstartview=FitH}
\usepackage{url}
\usepackage{xcolor}                % For colored text (optional)
\usepackage{lscape}                % For landscape pages if needed
\usepackage{caption}
\usepackage{subcaption}
\usepackage{etoolbox}
\AtBeginEnvironment{abstract}{\raggedright}
\usepackage{amsmath,amssymb,amsthm}
\usepackage[capitalise,nameinlink,noabbrev]{cleveref}
\crefname{equation}{Eq.}{Eqs.}
\Crefname{equation}{Eq.}{Eqs.}
\crefname{section}{Sec.}{Secs.}
\Crefname{section}{Section}{Sections}
\crefname{subsection}{Sec.}{Secs.}
\Crefname{subsection}{Section}{Sections}
\crefname{theorem}{Theorem}{Theorems}
\crefname{lemma}{Lemma}{Lemmas}
\crefname{proposition}{Proposition}{Propositions}
\crefname{corollary}{Corollary}{Corollaries}
\crefname{definition}{Definition}{Definitions}
\crefname{remark}{Remark}{Remarks}

\newtheorem{theorem}{Theorem}
\newtheorem{lemma}[theorem]{Lemma}
\newtheorem{proposition}[theorem]{Proposition}
\newtheorem{corollary}[theorem]{Corollary}
\newtheorem{definition}[theorem]{Definition}

\theoremstyle{remark}
\newtheorem{remark}[theorem]{Remark}

\newcommand{\OH}{\mathcal{H}_{\mathrm{OH}}}
\makeatletter
\g@addto@macro\normalsize{%
    \abovedisplayskip 3pt plus 1pt minus 1pt%
    \abovedisplayshortskip 3pt plus 1pt minus 1pt%
    \belowdisplayskip 3pt plus 1pt minus 1pt%
    \belowdisplayshortskip 3pt plus 1pt minus 1pt%
}
\makeatother

\def\BibTeX{{\rm B\kern-.05em{\sc i\kern-.025em b}\kern-.08em
    T\kern-.1667em\lower.7ex\hbox{E}\kern-.125emX}}

\author{Chinonso Onah}
\affiliation{Volkswagen AG, Berliner Ring 2, 38440 Wolfsburg, Germany}
\affiliation{Department of Physics, RWTH Aachen University, 52056 Aachen, Germany}
\email{conah@atom-computing.com}
\author{Kristel Michielsen}
\affiliation{Forschungszentrum Jülich, Germany}
\affiliation{Universit\"at zu K\"oln, 50923 K\"oln, Germany}

\title{Fundamental Limitations of QAOA on Constrained Problems and a Route  to Exponential Enhancement}

\begin{document}
\maketitle
% \title{}   % we won’t actually call \maketitle

\begin{abstract}
We study fundamental limitations of the generic Quantum
Approximate Optimization Algorithm (QAOA) on constrained
problems where valid solutions form a low-dimensional
manifold inside the Boolean hypercube, and we demonstrate a
provable route to exponential enhancements via constraint
embedding. Focusing on permutation-constrained objectives,
with ensembles of Travelling Salesperson Problems (TSP)
as representative examples, we show that the standard
``generic'' QAOA ansatz faces an intrinsic feasibility
bottleneck: even after parameter optimization over the
admissible domains, the feasible probability mass remains
exponentially suppressed by the Boolean embedding geometry. Our analysis triangulates this limitation through four
complementary lenses: (i) Walsh--Fourier/Krawtchouk
estimates showing that permutation one-hot constraints
have exponentially small low-degree and high-degree
spectral contributions, indicating that overcoming the
exponential geometric suppression requires substantial
contributions from the middle band. Probabilistic
estimates of objective-dependent interference establish
that the feasibility suppression persists with high
probability over the ensembles studied;
(ii) cost-angle averaging arguments that identify the
uniform feasible baseline $|\Pi|/2^N$ in the
nondegenerate case and characterize the surviving
interference through energy fibres;
(iii) fourth-moment ($L_4$) bounds that control
typical-angle concentration through energy-sum resonances;
and (iv) light-cone locality arguments that characterize
the propagation and overlap of constraint correlations,
with transition-feature estimates providing quantitative
control of feasibility in globally overlapping blocks. Against this envelope, we place a minimal constraint-enhanced kernel (CE--QAOA) that operates
directly inside a product one-hot subspace and mixes
with a block-local $XY$ Hamiltonian.
For permutation-constrained problems on $n$ objects
encoded in $N=n^2$ binary variables, at $p=1$, we prove,
with high probability over the instance, an angle-robust
exponential enhancement of the form
\begin{equation}
\label{eq:expp}
  \frac{P_{1}^{\mathrm{(CE)}}}{P_{1}^{\mathrm{(gen)}}}
  \;\ge\;
  \exp\!\bigl[\Omega(n^2)\bigr].
\end{equation}
Beyond one layer, for linked cost-phase schedules with
independently optimized mixer angles, the exponential
feasibility suppression persists with high probability
throughout linear-depth windows $p\leq\alpha n$,
for fixed $\alpha>0$.
\end{abstract}

% yielding an exponential
% feasibility enhancement across the entire depth window.
%These bounds further imply exponential raw-sampling requirements for generic QAOA to obtain even one feasible solution within the admissible parameter domains.

\section{Introduction}
\label{sec:introduction}

Combinatorial optimization problems with global constraints
pose a basic challenge for quantum sampling. A measured
configuration must satisfy the constraints before its objective
value makes it a useful candidate solution. For permutation
encodings on $N=n^2$ qubits, only $n!$ computational-basis
configurations are feasible. Uniform sampling from the Boolean
cube therefore produces a feasible configuration with probability
$n!/2^{n^2}$. Understanding how variational quantum dynamics
changes this probability is essential to assessing its
optimization performance. We study this question for the Quantum Approximate Optimization
Algorithm (QAOA). We use the term \emph{generic QAOA} for the
standard construction initialized in $|+\rangle^{\otimes N}$
and alternating diagonal cost evolution with a transverse-field
$X$ mixer. Our comparison is with constraint-enhanced QAOA
(CE--QAOA), which initializes in a product one-hot space and
preserves that space through a block-local $XY$ mixer. Both
constructions retain the optimization objective in their cost
evolution.

Permutation-constrained problems provide a natural setting for
this comparison. The Boolean encoding contains $2^{n^2}$
configurations, while the product one-hot space contains $n^n$.
The feasible permutation sector contains $n!$ configurations
in both cases. The travelling salesperson problem (TSP) and
its asymmetric variant realize this structure through their
city-position encoding. Their constraint penalties and
one-hot construction also establish their membership in the
CE--QAOA kernel. These shared problem instances allow us to
compare the two ansatz families through the probability of
sampling a feasible solution. The difference between the uniform feasible fractions motivates
the comparison, while the circuit dynamics determines whether
it persists after parameter optimization. We address this
through a probabilistic analysis of random edge-cost TSP
ensembles. The edge costs are sampled once and held fixed
throughout the circuit and parameter search. Our bounds hold
with high probability over this instance draw and simultaneously
over the parameter domains in the corresponding theorems.

At depth one, we prove that the generic feasible probability
remains $e^{-\Omega(n^2)}$ over all mixer angles and a
prescribed cost-angle grid. The grid may contain exponentially
many candidates in $n^2$, subject to the rate established in
the theorem. Because the bound holds simultaneously across
this domain, it also bounds the feasible probability attained
by an optimizer selecting parameters within it. The depth-one proof connects the harmonic structure of
permutation constraints with the objective-dependent phases.
A Walsh--Fourier/Krawtchouk analysis decomposes the feasible
probability into spectral contributions and controls the
low- and high-degree edges. Averaging over the edge-cost
ensemble expresses the remaining phase correlations through
coincidences of transition features. Anti-concentration and
rank estimates bound these collisions. Mixer-angle alignment
and a union bound then convert the ensemble estimate into
a simultaneous parameter-search bound.

At higher depth, the same objective appears in every layer,
so correlations between successive cost evolutions must be
retained. We introduce linked cost-angle schedules for which
the ensemble average admits a layerwise dephasing
representation. A bound on the sizes of transition-feature
fibres controls the averaged density operator, and
multivariate mixer-angle alignment extends this control to
all mixer angles. Together, these estimates establish
exponential feasibility suppression simultaneously throughout $p\le\alpha n$, for fixed $\alpha>0$, over the linked schedule
family. Our light-cone analysis identifies the supports and overlap
structure of the evolved constraint projectors.
Fubini--Tonelli and H\"older/Finner inequalities organize the
product-space averages and collision estimates, while
constraint-block bounds connect local support information
to joint feasibility. For the dense TSP encoding, the
transition-fibre estimate supplies the quantitative
suppression within the overlapping constraint block.

These simultaneous bounds have a direct sampling consequence.
They apply to adaptive procedures that choose each admissible
parameter setting using the instance and previous measurement
outcomes. If every admissible circuit has feasible probability
at most $u$, a procedure producing $S$ raw circuit samples
has probability at most $Su$ of observing a feasible sample
on the good-instance event. The sample budget includes
measurements used during training and those generated after
parameter selection. Exponentially small $u$ therefore
imposes an exponential sample requirement for
constant-probability feasible sampling. CE--QAOA provides an attainable comparison at every depth.
Setting its mixer angles to zero preserves the uniform
measurement probabilities on the product one-hot space
under arbitrary diagonal cost evolution, giving feasible
probability $n!/n^n$. Comparing this guarantee with the
generic upper bounds yields an $e^{\Omega(n^2)}$
enhancement throughout the studied depth window.
%The angle-agnostic existence guarantees inherited from the CE--QAOA kernel provide an additional comparison.

\medskip
\noindent
\subsection*{Paper Roadmap}
\pdfbookmark[2]{Paper Roadmap}{paper-roadmap}
% Replace the existing introductory roadmap beginning
% "The rest of the paper proceeds as follows" with this text.
% In the uploaded source, insert \label{sec:finitep} immediately after
% \subsection{Mixer optimization and the linear-depth window}.
% sec:finitep is the higher-depth section key supplied by the author.

\medskip
\noindent
The paper proceeds from the problem formulation to feasibility
bounds, their sampling consequences, and the mechanisms underlying
them. \Cref{sec:related} places the work in the literature
on variational trainability, locality, and constraint-preserving
ans\"atze. \Cref{sec:back} introduces generic QAOA and the
CE--QAOA kernel of \Cref{def:kernel-requirement}.
\Cref{sec:problem-formulation} specifies the permutation
constraints, the TSP and ATSP objectives, and the two circuit
families. In particular,
\Cref{sec:permutation-kernel-embedding} establishes the
kernel admissibility of the permutation encoding and its CE
feasibility baseline, while
\Cref{sec:formulation-circuit-comparison}
and~\ref{sec:formulation-tsp-ensemble} define the parameter
comparison and random edge-cost ensemble.

\Cref{sec:main} presents the main results.
The Walsh representation in \Cref{sec:walsh} connects
the mixer action in generic QAOA to the spectral structure of permutation
feasibility. \Cref{sec:upper} establishes the depth-one
obstruction in \Cref{thm:main-grid-feasibility}, and
\Cref{thm:p1-sep} compares it with the attainable CE
baseline to obtain exponential feasibility enhancement.
\Cref{sec:finite-grid-training} develops the consequences
for simultaneous parameter search, adaptive training, total
sampling cost, and angle resolution.
\Cref{sec:alternate} explains how parameter averaging,
fourth moments, and light-cone support provide complementary
views of the same concentration mechanism.
\Cref{sec:finitep} extends the parameter-uniform bound
to linked cost-angle schedules and independently variable
mixer angles, culminating in the linear-depth enhancement of
\Cref{cor:linked-linear-depth-separation}.
The Discussion relates these results to locality-based
limitations and the dense constraint geometry of the TSP
encoding. \Cref{sec:conclusion} concludes and identifies
further questions about parameter domains and coherent
amplitude transport as avenues for further work.

The derivations are collected in two appendices.
Appendix~\ref{sec:methods} develops the harmonic and
probabilistic analysis. Appendix~\ref{sec:kraw} derives the
Krawtchouk representation of Hamming-sphere indicators,
Appendix~\ref{sec:route1} develops the permutation-indicator
estimates, and Appendix~\ref{subsec:three-band-edges} bounds
the low- and high-degree spectral contributions.
Appendix~\ref{sec:probabilistic-p1-closure} proves the
depth-one ensemble estimate through Walsh-shell alignment,
a generating-function representation, and transition-collision
bounds, with the resulting shell suppression established in
Appendix~\ref{app:shell-suppression}.
Appendix~\ref{app:lenses} develops the complementary viewpoints.
Appendix~\ref{app:avg-angles} treats cost-angle averaging and
energy fibres, Appendix~\ref{app:L4} treats fourth moments
and energy-sum resonances, and Appendix~\ref{app:lightcone}
develops the light-cone, product-space, and transition-fibre
estimates that support the higher-depth extension.

\subsection{Related Work}
\label{sec:related}

Variational quantum algorithms (VQAs), and QAOA in particular, have been studied along three complementary axes: (i) \emph{trainability and landscape structure}, (ii) \emph{locality/light-cone limitations at shallow depth}, and (iii) \emph{architecture- or problem-aware ans\"atze} that encode constraints or symmetries from the outset. 
On the landscape side, early work identified \emph{barren plateaus} as a key obstruction to gradient-based training in large systems~\cite{McClean2018BarrenPlateaus,Cerezo2021VQAReview}. 
More recent analyses quantify landscape \emph{information content}, revealing how much task-relevant signal survives in variational loss surfaces~\cite{PerezSalinasWangBonetMonroig2024}, and study the \emph{energy landscape} structure of small-depth QAOA instances~\cite{choy2024energylandscapesquantumapproximate}. 
The QAOA design space has also broadened beyond the original transverse-field mixer to \emph{alternating operator ans\"atze} tailored to problem structure~\cite{HadfieldHoggRieffel2022}, while symmetry constraints can both obstruct and guide expressivity~\cite{Bravyi_2020}. 
Mitigation strategies such as \emph{classical shadows} for gradient estimation~\cite{SackMedinaMichailidisKuengSerbyn2022} and \emph{warm starts}~\cite{Egger2021Warm} aim to stabilize training or inject structure into initial states.

A second line of work formalizes \emph{locality} as the principal barrier at small depths since shallow circuits can only build correlations within a \emph{light cone} whose radius grows at most linearly with depth. These are essentially Lieb--Robinson-type locality bounds~\cite{Hastings2010Locality}, and they underlie modern limitations for low-depth QAOA on sparse graphs. For bounded-degree graph families, Farhi, Gamarnik, and Gutmann proved that small-$p$ QAOA ``needs to see the whole graph,'' producing typical- and worst-case examples where performance is capped by local neighborhoods~\cite{FarhiGamarnikGutmann2020WholeGraph}. At constant levels $p=O(1)$, quantitative barriers have been established on sparse hypergraphs and spin-glass models~\cite{Basso2022}, while in a broader framework of \emph{local quantum algorithms}, limitations propagate up to $p<\varepsilon\log n$ for certain CSPs with the (coupled) overlap-gap property~\cite{GamarnikZadik2022QAOAOGP,Anschuetz_2024}. 
Orthogonally, results on concentration for shallow circuits show output observables of low-depth evolutions are tightly concentrated, limiting global coordination at small $p$~\cite{BravyiGossetKonig2018}. Collectively, these works support the heuristic that constant or slowly growing depth cannot generate the long-range correlations required by \emph{global} constraints.

Recent work on Light Cone Cancellation (LCC) provides a particularly
direct link between these first two perspectives in
\emph{trainability and landscape structure} and
\emph{locality/light-cone limitations at shallow depth}~\cite{lee2026lightcone}.
By exploiting the causal cones of individual local observables, LCC
removes gates that cannot affect their expectation values, allowing the
variational objective to be evaluated using smaller subcircuits and
thereby reducing circuit resources and exposure to hardware noise. In this
sense, locality is used simultaneously as a tool for understanding limitations
of shallow variational circuits and as a resource for improving
their practical implementation.

A third, complementary, line of work focuses on \emph{architecture- and problem-aware}
ans\"atze that encode constraints, symmetries, or known structure directly at the
circuit level. In the QAOA setting, the alternating-operator framework extends the
original transverse-field mixer by replacing it with constraint-preserving unitaries
tailored to specific combinatorial structure~\cite{Hadfield2019AOA,HadfieldHoggRieffel2022}.
Constrained mixers and symmetry-projected variational families have been designed to
restrict the evolution to feasible manifolds or fixed symmetry sectors, thereby
reducing the effective search space and sometimes alleviating barren-plateau effects
~\cite{BaertschiEidenbenz2020,Fuchs2022ConstrainedMixers,Bravyi_2020}. Domain-specific encodings and mixers for routing and scheduling problems further illustrate how problem-aware design can unlock strong empirical performance on structured instances~\cite{Xie2024CVRP,tsvelikhovskiy2024symmetries}.

The present work connects these three perspectives under
a single quantitative lens by focusing on the feasible
probability mass generated by generic and constraint-enhanced
QAOA under global combinatorial constraints. We derive
high-probability upper bounds on the feasible probability
mass of generic QAOA for permutation-constrained problems,
using TSP/ATSP edge-cost ensembles and parameter domains
specified by the analysis.  Against this envelope, we place a minimal
constraint-enhanced kernel (CE--QAOA) that operates within
a product one-hot subspace and mixes with a block-local
$XY$ Hamiltonian, establishing an exponential separation
in feasible probability of the form
$\exp[\Omega(n^2)]$.
This comparison connects the harmonic, statistical, and
locality-based limitations of generic QAOA with the
advantages of problem--algorithm co-design developed in
the CE--QAOA kernel.

\subsection{Background}
\label{sec:back}

\subsubsection{Generic QAOA}
\label{subsec:generic-qaoa}
The generic QAOA~\cite{Farhi2014QAOA} prepares the state
\(
\ket{\psi_p(\vec\gamma,\vec\beta)}
=\bigl[\prod_{\ell=1}^{p}e^{-i\beta_\ell \sum_j X_j}\,e^{-i\gamma_\ell H_C}\bigr]\ket{+}^{\otimes N},
\)
i.e., a depth–$p$ alternation of a transverse–field $X$ mixer and a diagonal cost unitary on the computational basis. Parameters \((\vec\gamma,\vec\beta)\) are optimized (classically) to maximize an empirical objective. For the algorithm to be performant,  typical schedules need to be short (\(p=O(1)\) or \(p\ll N\)) without going to the adiabatic limit\cite{Lucas2014}. Common ideas to improve algorithmic performance include linear/adiabatic ramps\cite{montanezbarrera2024universalqaoa}, parameter transfer across sizes/instances to reduce training cost while keeping the same mixer/cost structures~\cite{Monta_ez_Barrera_2025T}, recursive elimination to iteratively fix variables based on low–depth runs, shrinking the instance before a final classical solve~\cite{Bravyi_2020, BaeLee2024RecursiveQAOA},  warmstarting QAOA to help landscapes and convergence~\cite{Egger2021Warm, carmo2025warmstarting}.

\subsubsection{Constraint–Enhanced QAOA}
\label{sec:ce-qaoa}

The \emph{alternating–operator} (``QAOA+'') paradigm augments the generic X–mixer
by feasibility–preserving mixers that act invariantly on constraint projectors,
thereby confining evolution to structured subspaces
\cite{Hadfield2019AOA,Fuchs2022ConstrainedMixers}. The
\emph{Constraint–Enhanced QAOA} (CE–QAOA) introduced in Ref. \cite{onahce} follows this direction but makes the problem–algorithm co–design explicit. As specified in \Cref{def:kernel-requirement},
the kernel defines the admissible problem classes,
penalty structure and symmetries, mixer, and initial state.

% We introduce a kernel designed to operate on block one–hot manifolds and \emph{fix} the mixer family and initial states to match the encoding. 

\begin{definition}[CE--QAOA kernel]
\label[definition]{def:kernel-requirement}
An optimization instance $I$ belongs to the \emph{CE--QAOA kernel} if there exist
integers $n,m\in\mathbb N$ and the \emph{one-hot} encoder $\mathsf E_{\mathrm{1hot}}$
that initializes the dynamics in the constrained fixed–Hamming–weight space
\[
\OH \;=\; (\mathcal H_1)^{\otimes m},
\qquad
\mathcal H_1 \;=\; \mathrm{span}\{\ket{e_1},\dots,\ket{e_n}\}
\quad\text{(one excitation per block)},
\]
together with Hamiltonians $H_{\mathrm{pen}}$ and $H_{\mathrm{obj}}$ on $\OH$
such that:
\begin{enumerate}[label=\textup{(\alph*)}, leftmargin=2.2em]
\item \emph{Penalty structure.} $H_{\mathrm{pen}}$ is a sum of squared affine
      one–hot/degree/capacity penalties (optionally plus linear forbids) with
      integer coefficients bounded by $\mathrm{poly}(n)$. Consequently,
      $\mathrm{spec}(H_{\mathrm{pen}})\subseteq\{0,1,\dots,t_{\max}\}$ with
      $t_{\max}=\mathrm{poly}(n)$.
\item \emph{Pattern symmetry.} $H_{\mathrm{pen}}$ is invariant under
      (i) block permutations $S_m$ and (ii) global symbol relabelings $S_n$.
      Hence the configuration space decomposes into level sets
      $L_t=\{x:\, H_{\mathrm{pen}}(x) :=\langle x \mid H_{\mathrm{pen}} \mid x \rangle \;=t\}$ that are preserved setwise.
\item \emph{Mixer and initial state.} The mixer family is block–local and normalized XY Hamiltonian,
      \[
      \widetilde H_{XY}^{(b)} \;=\;
      \frac{1}{n-1}\sum_{1\le a<b\le n}(X_aX_b+Y_aY_b),
      \qquad
      U_M(\beta) \;=\; \bigotimes_{b=1}^{m} e^{-i\beta\,\widetilde H_{XY}^{(b)}},
      \]
      with $\|\widetilde H_{XY}^{(b)}\|=O(1)$ on each block. The initial state is the
      uniform one–hot product
      \[
      \ket{s_0}\;=\;\ket{s_{\mathrm{blk}}}^{\otimes m},
      \qquad
      \ket{s_{\mathrm{blk}}}\;=\;\frac{1}{\sqrt n}\sum_{j=1}^n \ket{e_j}
      \quad\text{(a $W_n$ state initialized per block)}.
      \]
\end{enumerate}

\end{definition}

The identifiable problem classes are broad, including: Travelling Salesman (TSP/ATSP)~\cite{Lawler1985TSP,GareyJohnson1979}, Quadratic Assignment Problem (QAP)~\cite{Loiola2007QAPSurvey}, the shared transportation problems\cite{onah2025waas}, Generalized Assignment / Multiple–Knapsack~\cite{SahniGonzalez1976GAP}, $k$D Matching (NP-complete for $k\ge3$)~\cite{Karp1972}, etc. Some base cases are polynomial-time solvable but our interest is in the NP-hard variants (e.g. TSP/GAP/CVRP, etc.) for which the same one-hot encoding and symmetries apply.

\section{Problem formulation}
\label{sec:problem-formulation}

\subsection{Permutation-constrained optimization problems}
\label{sec:permutation-formulation}

%  The distinction from the present setting
% is that our quantity of interest is the total probability assigned to a
% globally constrained feasible sector, rather than a collection of local
% expectation values. This is also a key distinction between our approach
% and many analytical treatments of quantum optimization.

% These bounds hold uniformly
% over the admissible parameter searches at depth one and
% extend to linked cost-angle schedules throughout linear-depth
% windows.

We study permutation-constrained optimization problems and compare
the feasibility probabilities generated by generic QAOA and
constraint-enhanced QAOA (CE--QAOA). The travelling salesperson
problem provides our representative example~\cite{Lawler1985TSP,Lucas2014}: its encoding exposes
the permutation constraints, the objective-dependent phases, and
the transition geometry that enter the subsequent analysis.

Let $n\ge 2$ and write $[n]=\{1,\ldots,n\}$. We represent a
configuration by a binary matrix
\[
    x=(x_{ia})_{i,a\in[n]}\in\{0,1\}^{n\times n},
\]
where $x_{ia}=1$ assigns item $i$ to position $a$. We identify this matrix with its $N=n^2$-bit vectorization when using computational-basis states and Boolean Fourier coordinates. Throughout, lowercase $x$ (or $x^\prime$) denotes a binary configuration, while uppercase $X_{ia}$ and $Y_{ia}$ denote Pauli operators. The one-hot space is $\mathcal X$ as defined in \cref{eq:product-onehot-space}, and the feasible projector is $\widehat\Pi$ as defined in \cref{eq:formulation-feasible-probability}. Define the row and column occupations
\begin{equation}
    r_i(x)=\sum_{a=1}^n x_{ia},
    \qquad
    c_a(x)=\sum_{i=1}^n x_{ia}.
    \label{eq:permutation-occupations}
\end{equation}
The feasible set consists of permutation matrices,
\begin{equation}
    \Pi
    =
    \left\{
       x\in\{0,1\}^{n\times n}:
       r_i(x)=1,\ c_a(x)=1
       \text{ for every }i,a
    \right\},
    \qquad
    |\Pi|=n!.
    \label{eq:permutation-feasible-set}
\end{equation}
For an instance-dependent objective $f_D$, the optimization
problem is
\begin{equation}
    \min_{x\in\Pi} f_D(x).
    \label{eq:permutation-optimization}
\end{equation}

The standard quadratic constraint penalty is
\begin{align}
    V_{\mathrm{row}}(x)
       &=\sum_{i=1}^n \bigl(r_i(x)-1\bigr)^2,
       \nonumber\\
    V_{\mathrm{col}}(x)
       &=\sum_{a=1}^n \bigl(c_a(x)-1\bigr)^2,
       \nonumber\\
    V(x)&=V_{\mathrm{row}}(x)+V_{\mathrm{col}}(x).
    \label{eq:permutation-penalty}
\end{align}
Thus $V(x)\ge 0$ and $V(x)=0$ on $\Pi$. Given positive
penalty coefficients $A_{\mathrm r}$ and $A_{\mathrm c}$, define
\begin{equation}
    C_D(x)
    =
    f_D(x)
    +A_{\mathrm r}V_{\mathrm{row}}(x)
    +A_{\mathrm c}V_{\mathrm{col}}(x).
    \label{eq:permutation-total-cost}
\end{equation}
The corresponding diagonal Hamiltonian  is
\begin{equation}
    H_C(D)
    =
    \sum_{x\in\{0,1\}^{n\times n}}
       C_D(x)\,|x\rangle\langle x|.
    \label{eq:permutation-cost-hamiltonian}
\end{equation}
This specifies the objective and its extension to the full Boolean
space, including the energies of infeasible configurations.

We distinguish the full Boolean space from the product one-hot
space
\begin{equation}
    \mathcal X
    =
    \left\{
       x\in\{0,1\}^{n\times n}:r_i(x)=1
       \text{ for every }i
    \right\}
    \cong[n]^n,
    \qquad
    |\mathcal X|=n^n.
    \label{eq:product-onehot-space}
\end{equation}
Each element of $\mathcal X$ assigns one position to every item.
The remaining column constraints enforce distinct positions.
Consequently,
\(
    \Pi\subset\mathcal X\subset\{0,1\}^{N}.
\)
The corresponding uniform feasible fractions are
\begin{equation}
    \frac{|\Pi|}{2^N}=\frac{n!}{2^{n^2}},
    \qquad
    \frac{|\Pi|}{|\mathcal X|}=\frac{n!}{n^n}.
    \label{eq:two-feasibility-baselines}
\end{equation}
These fractions identify the geometric scales against which we
measure the effect of the circuit dynamics.

\subsection{TSP as a representative permutation problem}
\label{sec:tsp-formulation}

Consider a complete directed graph on $n$ cities with nonnegative
edge costs $D_{ij}$, $i\ne j$, and set $D_{ii}=0$. Rows of $x$
represent cities and columns represent cyclic tour positions.
With $a+1$ interpreted modulo $n$, the TSP objective in \cref{eq:permutation-optimization} is
\begin{equation}
    f_D(x)
    =
    \sum_{a=1}^n
    \sum_{\substack{i,j=1\\i\ne j}}^n
       D_{ij}\,x_{ia}x_{j,a+1}.
    \label{eq:tsp-objective}
\end{equation}
On $\Pi$, this is the length of the encoded Hamiltonian tour.
Independent directed costs give the asymmetric travelling
salesperson problem; imposing $D_{ij}=D_{ji}$ gives the symmetric
problem. We retain all $n!$ permutation encodings, including
equivalent cyclic representations of a tour. The penalized TSP cost is given by
\cref{eq:tsp-total-cost}, specializing
\cref{eq:permutation-total-cost} to the TSP objective
in the constrained optimization problem
\eqref{eq:permutation-optimization}.
\begin{equation}
    C_D(x)
    =
    f_D(x)
    +A_{\mathrm r}V_{\mathrm{row}}(x)
    +A_{\mathrm c}V_{\mathrm{col}}(x).
    \label{eq:tsp-total-cost}
\end{equation}
Replacing $x_{ia}$ by $(I-Z_{ia})/2$ produces a commuting diagonal
Hamiltonian containing at most two-qubit interactions. If $0\le D_{ij}\le D_{\max}$, every feasible tour has cost at most
$nD_{\max}$. Every infeasible binary configuration has a positive
integer constraint violation. Hence the sufficient choice
\begin{equation}
    \min\{A_{\mathrm r},A_{\mathrm c}\}>nD_{\max}
    \label{eq:tsp-sufficient-penalty}
\end{equation}
places every global minimizer of $C_D$ in $\Pi$. On $\Pi$, the
penalties vanish and minimization of $C_D$ agrees with
minimization of the tour length. The objective can also be expressed through its transition
features. Let $\mathsf S$ be the cyclic shift matrix with
$\mathsf S_{ab}=1$ when $b=a+1$ modulo $n$, and define
\begin{equation}
    M(x)=x\mathsf Sx^{\mathsf T},
    \qquad
    M_{ij}(x)=\sum_{a=1}^n x_{ia}x_{j,a+1}.
    \label{eq:tsp-transition-matrix}
\end{equation}
Then
\begin{equation}
    f_D(x)=\sum_{i\ne j}D_{ij}M_{ij}(x).
    \label{eq:tsp-transition-objective}
\end{equation}
For symmetric costs, the independent features are
$M_{ij}(x)+M_{ji}(x)$ with $i<j$. On a feasible configuration, these features record the tour
edges. On the full Boolean space, they record the transitions
generated by its occupied positions. Their coincidences govern
which pairs of configurations retain coherence under the
objective averages. % used later in the paper.

\subsection{Embedding in the CE--QAOA kernel}
\label{sec:permutation-kernel-embedding}

Following \Cref{def:kernel-requirement}, group the qubits into $n$ row blocks and define
\begin{equation}
    \mathcal H_{\mathrm{OH}}
    =
    \bigotimes_{i=1}^n
       \operatorname{span}\{|e_1\rangle_i,\ldots,|e_n\rangle_i\}
    =
    \operatorname{span}\{|x\rangle:x\in\mathcal X\}.
    \label{eq:permutation-onehot-hilbert-space}
\end{equation}
The initial state is the uniform product one-hot state (a $W-$state),
\begin{equation}
    |s_0\rangle
    =
    |W_n\rangle^{\otimes n},
    \qquad
    |W_n\rangle
    =
    \frac{1}{\sqrt n}\sum_{a=1}^n|e_a\rangle.
    \label{eq:permutation-ce-initial-state}
\end{equation}
Within row block $i$, use the normalized complete-graph
$XY$ mixer
\begin{equation}
    \widetilde H_{XY}^{(i)}
    =
    \frac{1}{n-1}
    \sum_{1\le a<b\le n}
       \bigl(X_{ia}X_{ib}+Y_{ia}Y_{ib}\bigr),
    \qquad
    U_M(\beta)
    =
    \bigotimes_{i=1}^n
       e^{-i\beta\widetilde H_{XY}^{(i)}}.
    \label{eq:permutation-ce-mixer}
\end{equation}
Each block mixer preserves its excitation number. On the
one-excitation sector,
\[
    \widetilde H_{XY}^{(i)}
    =
    \frac{2}{n-1}(J_n-I_n),
    \qquad
    \left\|
       \widetilde H_{XY}^{(i)}
       \big|_{\mathcal H_1}
    \right\|=2,
\]
where $J_n$ is the all-ones matrix. Diagonal cost evolution also
preserves $\mathcal H_{\mathrm{OH}}$. Every CE--QAOA layer
therefore acts within this encoded space.

On $\mathcal X$, the row penalty vanishes identically and the
remaining constraint penalty is
\begin{equation}
    V(x)\big|_{\mathcal X}
    =
    V_{\mathrm{col}}(x)
    =
    \sum_{a=1}^n c_a(x)^2-n.
    \label{eq:encoded-permutation-penalty}
\end{equation}
It is invariant under permutations of the row blocks and under
a common permutation of the position labels. Its zero level
set is precisely $\Pi$, and
\begin{equation}
    V_{\mathrm{col}}(\mathcal X)
    \subseteq\{0,1,\ldots,n(n-1)\}.
    \label{eq:encoded-penalty-spectrum}
\end{equation}
The upper endpoint is attained when all rows occupy the same column. Thus, TSP and ATSP are kernel-admissible problems under the
specified permutation encoding, satisfying the CE-QAOA kernel definition in \Cref{def:kernel-requirement}. They therefore inherit the established CE--QAOA kernel guarantees from Ref. \cite{onahce} with $m=n$. This identification allows the CE--QAOA guarantees established
for the kernel to be specialized to $m=n$. In particular,
\Cref{thm:exist-params} supplies its angle-agnostic
existence bound used in subsequent analyses.

\subsection{Circuit comparison and parameter search}
\label{sec:formulation-circuit-comparison}

For the same instance Hamiltonian $H_C(D)$ (Eq. \ref{eq:permutation-cost-hamiltonian}), generic QAOA uses
the initial state $|+\rangle^{\otimes N}$ and the transverse
mixer
\[
    U_X(\beta)
    =
    \exp\left(-i\beta\sum_{i,a}X_{ia}\right).
\]
Writing layer $1$ as the rightmost factor, the two depth-$p$
states are
\begin{align}
    |\psi_p^{(\mathrm{gen})}\rangle
    &=
    U_X(\beta_p)e^{-i\gamma_pH_C(D)}
    \cdots
    U_X(\beta_1)e^{-i\gamma_1H_C(D)}
    |+\rangle^{\otimes N},
    \nonumber\\
    |\psi_p^{(\mathrm{CE})}\rangle
    &=
    U_M(\beta_p)e^{-i\gamma_pH_C(D)}
    \cdots
    U_M(\beta_1)e^{-i\gamma_1H_C(D)}
    |s_0\rangle.
    \label{eq:formulation-two-ansatze}
\end{align}
Define the feasible projector and output probabilities by
\begin{equation}
    \widehat\Pi
    =
    \sum_{x\in\Pi}|x\rangle\langle x|,
    \qquad
    P_{\Pi,p}^{(A)}(\boldsymbol\gamma,\boldsymbol\beta;D)
    =
    \langle\psi_p^{(A)}|
       \widehat\Pi
    |\psi_p^{(A)}\rangle,
    \quad
    A\in\{\mathrm{gen},\mathrm{CE}\}.
    \label{eq:formulation-feasible-probability}
\end{equation}
The objective remains present throughout the evolution.
Feasibility measures the probability that a raw circuit
sample is a valid candidate for the optimization problem.
For any set of acceptable solutions $\mathcal T\subseteq\Pi$,
\[
    \Pr(x\in\mathcal T)
    \le P_{\Pi,p}^{(A)}.
\]
For a prescribed parameter domain $\Theta_{n,p}$, we define
the optimized feasible mass
\begin{equation}
    P_{\Pi,p}^{(A),\star}(D;\Theta_{n,p})
    =
    \sup_{\boldsymbol\theta\in\Theta_{n,p}}
       P_{\Pi,p}^{(A)}(\boldsymbol\theta;D).
    \label{eq:formulation-optimized-feasible-mass}
\end{equation}
In what follows, each subsequent theorem specifies its parameter domain and
depth range. A uniform estimate on this supremum applies to
every parameter-selection procedure whose outputs belong to
that domain, including adaptive procedures using previous
circuit measurements.

\subsection{Instance ensemble for the probabilistic analysis}
\label{sec:formulation-tsp-ensemble}

For the probabilistic results, take the independent edge
costs to be uniform on
\[
    \{0,1,\ldots,Q-1\},
\]
where $Q$ is prime. In the asymmetric model, the independent
variables are $D_{ij}$ for $i\ne j$. In the symmetric model,
they are $D_{ij}$ for $i<j$, with $D_{ji}=D_{ij}$.
The penalty coefficients are deterministic functions of
$n$ and $Q$; for example,
\[
    A_{\mathrm r}=A_{\mathrm c}=n(Q-1)+1
\]
satisfies \cref{eq:tsp-sufficient-penalty}.  A sampled
instance is held fixed throughout every circuit layer and
throughout parameter search. At depth one, the cost-angle grid is
\begin{equation}
    \Gamma_Q
    =
    \left\{\frac{2\pi a}{Q}:a=0,\ldots,Q-1\right\}.
    \label{eq:formulation-cost-angle-grid}
\end{equation} This connects the objective ensemble to phase
orthogonality. For $a\ne0$ and integer $k$,
\begin{equation}
    \frac1Q\sum_{d=0}^{Q-1}
       e^{-2\pi i a d k/Q}
    =
    \mathbf 1_{\{k\equiv0\pmod Q\}}.
    \label{eq:formulation-phase-orthogonality}
\end{equation}
Here $a$ indexes nonzero cost angles
in $\Gamma_Q$, $d$ denotes a possible value of an
independently sampled edge cost, and $k\in\mathbb Z$
is an integer coefficient arising from differences
of transition features. The indicator on the right
equals one when $k$ is divisible by $Q$ and zero
otherwise. Since $Q$ is prime and $a\ne0$, the sum
vanishes unless $k\equiv0\pmod Q$. Together with the transition representation
in \cref{eq:tsp-transition-objective}, this identity converts
objective-phase averages into constraints on transition
features. The harmonic and light-cone analyses then relate
these constraints to feasible output probability. % and its dependence on depth.

% The resulting comparison combines three explicit components:
% the permutation feasibility geometry, the TSP objective
% features, and the parameter domain accessible to the search.
% The CE--QAOA kernel supplies the corresponding feasibility
% guarantees within the product one-hot space.

\section{Main Results}
\label{sec:main}
\subsection{Walsh--Fourier/Krawtchouk Analysis}

We analyse the one-layer generic ansatz from
\cref{sec:formulation-circuit-comparison} using the permutation
encoding of \cref{sec:permutation-formulation} and the ensemble specification of \cref{sec:formulation-tsp-ensemble}.
For the diagonal cost $C(y)$, write $\phi_\gamma(y)=e^{-i\gamma C(y)}$
for its phase field. The computational-basis kernel of the transverse-field
mixer factors over the $N=n^2$ qubits as
\begin{equation}
 K_\beta(x,y)
 =\langle x|U_X(\beta)|y\rangle
 =\prod_{u=1}^N\langle x_u|e^{-i\beta X}|y_u\rangle
 =(\cos\beta)^{N-d(x,y)}(-i\sin\beta)^{d(x,y)},
 \label{eq:kb}
\end{equation}
where $d(x,y)=|x\oplus y|$ is the Hamming distance. Thus
\begin{equation}
 a_{\beta,\gamma}(x)
 =2^{-N/2}\sum_y K_\beta(x,y)\phi_\gamma(y).
 \label{eq:prob_def}
\end{equation}
Unitarity of $K_\beta$ and $|\phi_\gamma(y)|=1$ give
$\sum_x|a_{\beta,\gamma}(x)|^2=1$.

\subsection{Walsh--Fourier primer}
\label{sec:walsh}

For a mask $S\subseteq[N]$, identify $S$ with its binary indicator
vector and write $S\cdot x=\sum_{u\in S}x_u\bmod2$.
The Walsh characters are $\chi_S(x)=(-1)^{S\cdot x}$
\cite{ODonnell2014}. Symmetric difference $S\oplus T$ is the group
operation on masks, and
$\chi_S(x)\chi_T(x)=\chi_{S\oplus T}(x)$.
We use the normalized transform
\begin{equation}
 \widehat f(S)=2^{-N}\sum_x f(x)\chi_S(x),
 \qquad f(x)=\sum_{S\subseteq[N]}\widehat f(S)\chi_S(x).
 \label{eq:ft}
\end{equation}
For the counting inner product
$\langle f,g\rangle=\sum_x\overline{f(x)}g(x)$, Plancherel gives
\[
 2^{-N}\langle f,g\rangle
 =\sum_S\overline{\widehat f(S)}\widehat g(S),
 \qquad
 \sum_x|f(x)|^2=2^N\sum_S|\widehat f(S)|^2.
\]
We denote the normalized inner product by
$\langle\!\langle f,g\rangle\!\rangle=2^{-N}\langle f,g\rangle$. The kernel in \cref{eq:kb} acts by convolution on $\mathbb Z_2^N$,
\[
 (K_\beta f)(x)=\sum_z k_\beta(z)f(x\oplus z),
 \qquad
 k_\beta(z)=\prod_{u=1}^N\kappa_\beta(z_u),
\]
where $\kappa_\beta(0)=\cos\beta$ and
$\kappa_\beta(1)=-i\sin\beta$.

\begin{lemma}[Diagonalization of $K_\beta$ in the Walsh basis]
\label[lemma]{lem:walsh-diag}
For every $S\subseteq[N]$,
\[
 \widehat{K_\beta f}(S)=\lambda_\beta(|S|)\widehat f(S),
 \qquad \lambda_\beta(|S|)=e^{-i\beta(N-2|S|)}.
\]
Consequently, $K_\beta$ preserves each Walsh magnitude
$|\widehat f(S)|$.
\end{lemma}

\begin{proof}
The one-bit characters have mixer eigenvalues
$\cos\beta-i\sin\beta=e^{-i\beta}$ and
$\cos\beta+i\sin\beta=e^{i\beta}$.
Tensorization gives $e^{-i\beta(N-|S|)}e^{i\beta|S|}$ for mask $S$,
which is $\lambda_\beta(|S|)$.
\end{proof}

\paragraph{Permutation-feasibility indicator.}
Use the encoding from \cref{sec:tsp-formulation}, with rows indexed
by cities and columns by tour positions. Let
\[
 R(x)=\prod_{i=1}^n\mathbf1_{\{\sum_{j=1}^n x_{ij}=1\}},
 \qquad
 I_{\mathrm{col}}(x)=\prod_{j=1}^n\mathbf1_{\{\sum_{i=1}^n x_{ij}=1\}}
\]
be the row and column one-hot indicators. Then
$\mathbf1_\Pi=R I_{\mathrm{col}}$ and
$\|\mathbf1_\Pi\|_2^2=|\Pi|=n!$. Writing $g=K_\beta\phi_\gamma$, \cref{eq:prob_def} and
Plancherel express feasible mass as
\[
 P_\Pi(\beta,\gamma)
 =2^{-N}\sum_x\mathbf1_\Pi(x)|g(x)|^2
 =\sum_S\widehat{\mathbf1_\Pi}(S)\widehat{|g|^2}(S).
\]
For $b(T)=\widehat{\phi_\gamma}(T)$,
\Cref{lem:walsh-diag} and the Walsh product identity give
\[
 \widehat{|g|^2}(S)
 =\sum_T\lambda_\beta(|T|)\overline{\lambda_\beta(|T\oplus S|)}
 b(T)\overline{b(T\oplus S)},
 \qquad \sum_T|b(T)|^2=1.
\]
The feasible mass therefore couples the Fourier profile of the
permutation indicator to Walsh autocorrelations of the phase field.
Appendix~\ref{sec:kraw} develops the Krawtchouk representation and
spectral estimates used in this analysis.

\subsection{Fundamental Limitations}
\label{sec:upper}

Our point of departure is the following question. By how much
can a shallow variational quantum circuit concentrate probability
mass on the feasible manifold of a globally constrained problem?
We answer this for permutation-constrained problems one-hot
encoded into $N=n^2$ qubits, using random edge-cost TSP/ATSP
ensembles. We prove that generic QAOA produces exponentially
small feasible probability with high probability over the
instance, uniformly over the parameter domains considered.
At depth one, these domains include all mixer angles and
cost-angle grids of up to exponential cardinality. At higher
depth, the bound persists throughout $p\le\alpha n$, for fixed
$\alpha>0$, over linked cost-angle schedules and all mixer
angles. The feasible sector contains $|\Pi|=n!$ configurations.

We quantify the feasibility bottleneck through a bound that
holds simultaneously over the parameters available to the
search. Consider the TSP/ATSP ensemble introduced in
\Cref{sec:formulation-tsp-ensemble}, with independent
edge costs uniform on $\{0,\ldots,Q-1\}$, where $Q>2n$ is
prime. The constraint-penalty coefficients are deterministic
functions of $n$ and $Q$. Each sampled instance $D$ is held
fixed throughout parameter selection and circuit execution. For the depth-one generic circuit, write
\begin{equation}
    P_{\Pi,1}^{(\mathrm{gen}),\star}(D;\Gamma_Q)
    =
    \sup_{\substack{\beta\in\mathbb R\\
                    \gamma\in\Gamma_Q}}
       P_{\Pi,1}^{(\mathrm{gen})}(\beta,\gamma;D).
    \label{eq:main-grid-optimized-feasibility}
\end{equation}
This quantity includes optimization over every mixer angle
and every cost-angle candidate in the grid defined in
\cref{eq:formulation-cost-angle-grid}.

\begin{theorem}[Depth-one feasibility obstruction]
\label[theorem]{thm:main-grid-feasibility}
There is a constant $c_\ast>0$, independent of $n$ and $Q$,
such that, for all sufficiently large $n$, with probability
at least
\[
    1-\delta_{n,Q},
    \qquad
    \delta_{n,Q}
    =
    \min\left\{
       1,\,
       (Q-1)e^{-c_\ast n^2/2}
    \right\},
\]
over the edge-cost instance,
\begin{equation}
    P_{\Pi,1}^{(\mathrm{gen}),\star}(D;\Gamma_Q)
    \le
    u_n,
    \qquad
    u_n=(n^2+1)n!\,e^{-c_\ast n^2/2}.
    \label{eq:main-grid-feasibility-bound}
\end{equation}
The constant may be chosen with $c_\ast<\log 2$.
In particular, if
\begin{equation}
    \log Q
    \le
    \left(\frac{c_\ast}{2}-\varepsilon\right)n^2,
    \qquad
    0<\varepsilon<\frac{c_\ast}{2},
    \label{eq:main-admissible-grid-size}
\end{equation}
then $\delta_{n,Q}\le e^{-\varepsilon n^2}$ and
\[
    u_n
    =
    \exp\left[
       -\frac{c_\ast}{2}n^2+O(n\log n)
    \right]
    =
    e^{-\Omega(n^2)}.
\]
\end{theorem}

The proof combines three estimates developed in detail in App. \ref{sec:probabilistic-p1-closure} culminating in the uniform shell-coherence suppression in App. \ref{app:shell-suppression}. At the core are the following:
First, mixer-angle averaging separates the Walsh-degree
contributions (\Cref{lem:single-layer-shell-alignment} and
\Cref{cor:full-feasible-alignment}) and gives
\begin{equation}
    \sup_\beta
       P_{\Pi,1}^{(\mathrm{gen})}(\beta,\gamma;D)
    \le
    (N+1)\,
       \overline P_{\Pi,1}^{(\mathrm{gen})}(\gamma;D),
    \label{eq:main-mixer-alignment}
\end{equation}
where
\[
    \overline P_{\Pi,1}^{(\mathrm{gen})}(\gamma;D)
    =
    \frac1\pi\int_0^\pi
       P_{\Pi,1}^{(\mathrm{gen})}(\beta,\gamma;D)
       \,d\beta.
\]
Second, the transition-collision bound of
\Cref{lem:transition-collision} and
\cref{eq:expected-dephased-feasible} yield
\begin{equation}
    \mathbb E_D
       \overline P_{\Pi,1}^{(\mathrm{gen})}
          \left(\frac{2\pi a}{Q};D\right)
    \le
    n!e^{-c_\ast N},
    \qquad
    a=1,\ldots,Q-1.
    \label{eq:main-ensemble-feasibility-mean}
\end{equation}
Markov's inequality at threshold
$n!e^{-c_\ast N/2}$, followed by a union bound over the
$Q-1$ cost angles, establishes
\cref{eq:main-grid-feasibility-bound}.
At $\gamma=0$, the output distribution is uniform and
the feasible probability is $n!/2^N$, which also satisfies
the bound. The harmonic analysis resolves the feasible probability
into low-degree, middle-degree, and high-degree
contributions. Deterministic estimates control the
spectral edges. The transition-feature collision
estimates supply the ensemble control needed for the
remaining contribution. This produces a bound on the
full feasible probability, including the
objective-dependent interference.

\Cref{eq:main-admissible-grid-size} permits both
polynomial grids and grids whose cardinality grows
exponentially with $n^2$. Here $Q$ specifies the edge-cost
ensemble as well as the cost-angle grid. The theorem
therefore describes a family of instance ensembles and
associated search domains. Because the bound is simultaneous, parameter selection
may depend on the sampled instance and on previous
measurement outcomes. On the good-instance event, each
raw circuit sample has conditional feasible probability
at most $u_n$. 
 %Constant-probability feasible sampling therefore requires $S=e^{\Omega(n^2)}$ in the regime of \cref{eq:main-admissible-grid-size}.

% The higher-depth analysis establishes a corresponding
% simultaneous bound over linked cost-angle schedules
% and all mixer angles throughout $p\le\alpha n$.
% The CE--QAOA comparison below turns these generic
% feasibility bounds into exponential enhancement
% guarantees.

% Macros
% ===================== Setup =====================
\newcommand{\Hone}{\mathcal H_1}
\newcommand{\Htot}{\mathcal H}

\newcommand{\E}{\mathbb{E}}

\subsection{The Path to Exponential Enhancements}
%==================================================
% Depth-1 grid search on wi4 / wi5 (Aer, noiseless)
%==================================================
\begin{figure*}[t]
  \centering
  %--------------------------- wi4 ---------------------------
  \begin{subfigure}[t]{0.4\linewidth}
    \centering
    \includegraphics[width=\linewidth]{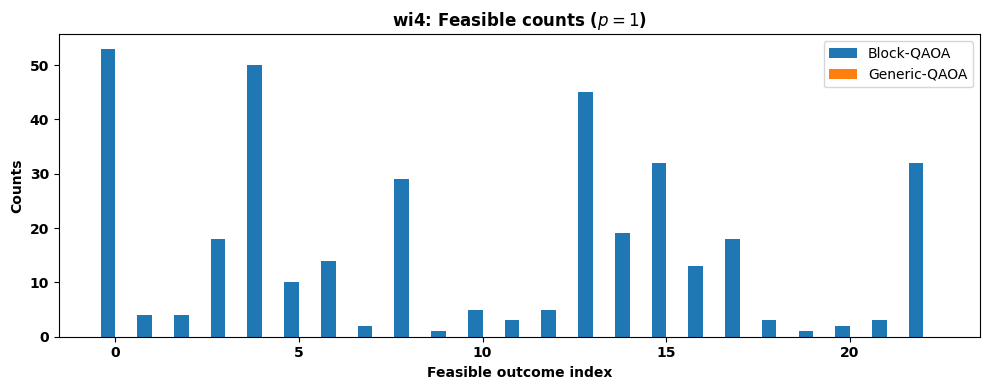}
    \caption{\textbf{wi4, $n{=}4$ cities (16 qubits).}
             Blue bars: Block-QAOA ($p{=}1$) with \((\gamma,\beta)\) taken
             from the $10\times10$ grid.
             Orange bars (invisible) would correspond to generic QAOA,
             but every feasible bit-string occurred \emph{zero} times.}
    \label{fig:wi4_counts}
  \end{subfigure}
  %--------------------------- wi5 ---------------------------
  \begin{subfigure}[t]{0.4\linewidth}
    \centering
    \includegraphics[width=\linewidth]{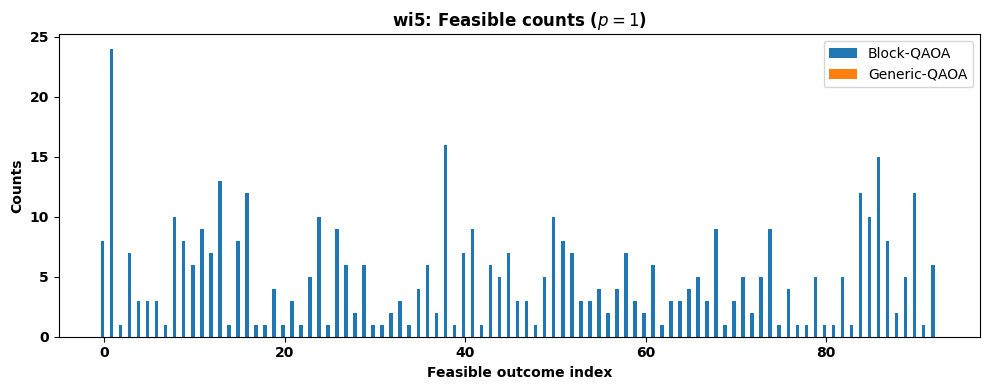}
    \caption{\textbf{wi5, $n{=}5$ cities (25 qubits).}
             Same colour code as (a).
             Again, generic QAOA produced no feasible outputs,
             while the constrained ansatz covers the full set of
             \(5!\!=\!120\) tours; the best tour appeared \(24\) times.}
    \label{fig:wi5_counts}
  \end{subfigure}
  \vspace{-0.7em}
  \caption{\textbf{Constraint-Enhanced QAOA dominates generic QAOA at
           depth \(p{=}1\).}   For medium-size TSP instances from QOptLib\cite{Osaba2024Qoptlib}, the unconstrained ansatz       spreads its amplitude almost entirely outside the            permutation subspace (\(n!/2^{n^{2}}\!\to\!0\)),           whereas the block-encoded circuit concentrates an
           \(\Omega(n^{-m})\) fraction of its probability on valid tours.}
  \label{fig:wi4_wi5_grid}
\end{figure*}

\begin{figure}[t]
  \centering
  \includegraphics[width=.65\linewidth]{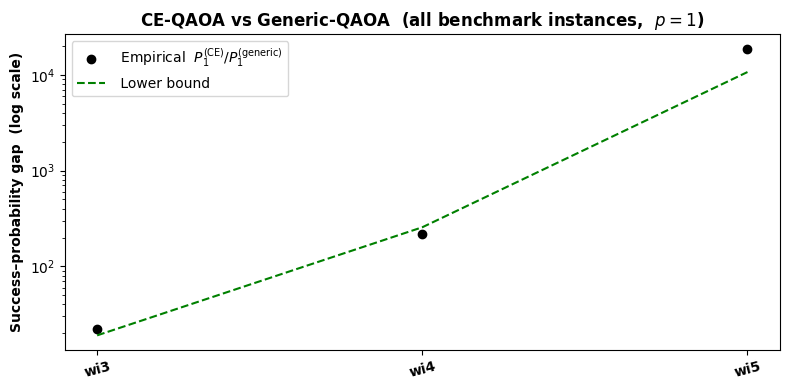}
  \caption{Empirical versus theoretical feasibility mass after parameter optimization on noiseless simulations at depth p = 1 for the three QOPLib benchmark instances wi3, wi4, and wi5 \cite{Osaba2024Qoptlib}. Each instance was executed with 500,000 shots so that Generic QAOA returns at least one feasible bit-string}.
    \label{fig:Blockqaoa_feasibility}
\end{figure}

The constraint-enhanced QAOA (CE--QAOA) operates
within the product one-hot space and mixes within
each block using an all-to-all \(XY\) mixer.
This encoding supplies a feasible-probability
baseline that yields an \(e^{\Omega(n^2)}\)
enhancement over generic QAOA on the stated
TSP/ATSP ensembles and parameter domains.
The comparison holds at depth one and throughout
the linear-depth window established below.

In order to make statements regarding the ratio $
\frac{P^{\mathrm{(CE)}}_{p}}{P^{\mathrm{(generic)}}_{p}}$, we need some robust lower envelope for CE--QAOA which we can compare to the upper envelope for generic QAOA outlined in Sec \ref{sec:upper}.  The necessary lower envelope comes from a basic symmetry of the product one-hot construction in Sec. \ref{sec:ce-qaoa}. For an instance in the kernel with $m$ blocks of local dimension $n$, a depth–$p$ CE–QAOA layer stack is given as 
\[
  \ket{\psi_p(\vec\gamma,\vec\beta)}
  \;=\;
  \Bigl(\prod_{\ell=1}^{p} U_M(\beta_\ell)\,e^{-i\gamma_\ell H_C}\Bigr)\ket{s_0},
  \qquad
  \vec\gamma=(\gamma_1,\dots,\gamma_p),\;
  \vec\beta=(\beta_1,\dots,\beta_p).
\]
Because $U_M$ preserves the one–hot sector and $H_C$ is diagonal, each layer maps the encoded manifold $\OH=(\mathcal H_1)^{\otimes m}$ to itself. The kernel’s pattern symmetry (block permutations $S_m$ and global symbol
relabelings $S_n$) permits a \emph{block–permutation twirl}. Averaging any fixed circuit $U$ over such relabelings yields a robust lower bound on the probability profile of the prepared state. For $p=1$, let \(\ket{x^\star}\in \OH\) be any fixed product basis vector and $n^m$ be the size of the product one-hot manifold explored by CE-QAOA, we recall the following lower bounds from Ref. \cite{onahce}:

 \begin{theorem}[Angle Agnostic Existence Bound \cite{onahce}] \label[theorem]{thm:exist-params} For any product basis vector $\ket{x^\star}\in \OH$, and any angles \((\gamma,\beta)\), there exists a blockwise permutation \(\mathsf P^\star\) such that \[ \bigl|\langle x^{\star}\!\mid \mathsf P^{\star\dagger} U_{M}(\beta) U_{C}(\gamma) |s_{0}\rangle \bigr|^{2} \;\ge\; \frac{c}{n^{m}}. \] \end{theorem}

 \newcommand{\UM}{U_{M}}
\newcommand{\UC}{U_{C}}

 % \begin{corollary}[Feasible–optimum specialization of Thm.~\ref{thm:exist-params}\cite{onahce}]
 % \label[corollary]{cor:feasible-optimum} Fix any grid angles $(\gamma,\beta)\in\mathcal G_{n+1}\times\mathcal G_{n+1}$ and set $U=\UM(\beta)\UC(\gamma)$ on $\OH$ with $D=\dim\OH=n^m$. Let $x^\star$ be any \emph{feasible optimal} label. Then there exists a blockwise permutation $\mathsf P^\star$ such that \[ \bigl|\langle x^\star \mid \mathsf P^{\star\dagger} U \mid s_0 \rangle\bigr|^2 \;\ge\; \frac{1}{D} \;=\; \frac{1}{n^m}. \] In particular, the constant in \Cref{thm:exist-params} can be taken as $c=1$. \end{corollary}

\medskip
\noindent

From the foregoing, we can therefore conclude that for every target basis
state \(x^{\star}\in\OH\), and in particular for every feasible optimal label,
the blockwise relabeling symmetry guarantees the existence of a blockwise
permutation \(\mathsf P^\star\) such that, for the fixed grid angles
\((\gamma,\boldsymbol\beta)\) under consideration (Corollary ~16 in \cite{onahce}),
\begin{equation}
  \bigl|\langle x^{\star}\!\mid
          \mathsf P^{\star\dagger}
          U_{M}(\boldsymbol\beta)
          U_{C}(\gamma)
          |s_{0}\rangle
  \bigr|^{2}
  \;\ge\;
  \frac{1}{n^{m}}.
  \label{eq:lb}
\end{equation}
For permutation-constrained problems with \(m=n\), this gives the lower
envelope
\[
  P_{1}^{\mathrm{(CE)}} \;\ge\; n^{-n}.
\]

The product one-hot initialization supplies a feasible-probability
baseline of \(n!/n^n\). Combining this attainable baseline with
Theorem~\ref{thm:main-grid-feasibility} yields the following result.

\begin{theorem}[Single-layer exponential enhancement]
\label[theorem]{thm:p1-sep}
Consider the TSP/ATSP ensemble and cost-angle grid
$\Gamma_Q$ of \Cref{thm:main-grid-feasibility}.
Define the optimized CE--QAOA feasible probability on
the same parameter domain by
\[
    P_{\Pi,1}^{(\mathrm{CE}),\star}(D;\Gamma_Q)
    =
    \sup_{\substack{\beta\in\mathbb R\\
                    \gamma\in\Gamma_Q}}
       P_{\Pi,1}^{(\mathrm{CE})}(\beta,\gamma;D).
\]
With probability at least $1-\delta_{n,Q}$ over the instance,
\begin{equation}
    \frac{
       P_{\Pi,1}^{(\mathrm{CE}),\star}(D;\Gamma_Q)
    }{
       P_{\Pi,1}^{(\mathrm{gen}),\star}(D;\Gamma_Q)
    }
    \ge
    \frac{e^{c_\ast n^2/2}}{(n^2+1)n^n}
    =
    \exp\left[
       \frac{c_\ast}{2}n^2-n\log n-\log(n^2+1)
    \right].
    \label{eq:p1-probabilistic-enhancement}
\end{equation}
Consequently, whenever
\[
    \log Q
    \le
    \left(\frac{c_\ast}{2}-\varepsilon\right)n^2
\]
for fixed $0<\varepsilon<c_\ast/2$, the enhancement is
$e^{\Omega(n^2)}$ with probability at least
$1-e^{-\varepsilon n^2}$.
\end{theorem}

\begin{proof}
The CE--QAOA initialization assigns probability $n^{-n}$
to each configuration in the product one-hot space.
Setting $\beta=0$ preserves these probabilities under
the diagonal cost evolution for every $\gamma\in\Gamma_Q$.
Since $n!$ configurations are feasible,
\[
    P_{\Pi,1}^{(\mathrm{CE}),\star}(D;\Gamma_Q)
    \ge \frac{n!}{n^n}.
\]
On the good-instance event of
\Cref{thm:main-grid-feasibility},
\[
    P_{\Pi,1}^{(\mathrm{gen}),\star}(D;\Gamma_Q)
    \le
    (n^2+1)n!e^{-c_\ast n^2/2}.
\]
Dividing establishes
\cref{eq:p1-probabilistic-enhancement}.
The probability estimate follows from the 
bound on $Q$.
\end{proof}

\noindent
\Cref{thm:p1-sep} quantifies the feasibility benefit
of constraint embedding on the same instances and parameter
domain. The product one-hot initialization supplies an
attainable CE--QAOA baseline, while the generic upper bound
holds throughout parameter optimization as we discuss next.

\subsection{Finite-grid parameter search, adaptive training,
and sampling complexity}
\label{sec:finite-grid-training}

As an algorithmic consequence, the depth-one feasibility bound in \Cref{thm:main-grid-feasibility} controls the entire admissible parameter set simultaneously.
Consequently, an optimizer cannot evade the bound by choosing its
parameters adaptively, provided its returned cost angles belong to
the prescribed grid. This permits a direct accounting of circuit
samples used during training and subsequent sampling. The Walsh-shell alignment bound of
\Cref{cor:full-feasible-alignment} gives, pointwise in \(D\) and
\(\gamma\),
\begin{equation}
\sup_{\beta\in\mathbb R}
P_\Pi^{(\mathrm{gen})}(\beta,\gamma;D)
\le
(N+1)\overline P_\Pi(\gamma;D).
\label{eq:training-alignment-input}
\end{equation}
Decreasing \(c_*\), if necessary, we may assume
\(0<c_*<\log 2\).

The preceding collision argument in \cref{eq:main-ensemble-feasibility-mean} supplies a constant \(c_*>0\),
independent of \(n\) and \(Q\), such that, for every nonzero
\(\gamma_a\in\Gamma_Q\),
\begin{equation}
\mathbb E_D\,
\overline P_\Pi(\gamma_a;D)
\le
n!e^{-c_*n^2},
\label{eq:training-dephased-input}
\end{equation}
where
\(
\overline P_\Pi(\gamma;D)
=
\frac{1}{\pi}
\int_0^\pi
P_\Pi^{(\mathrm{gen})}(\beta,\gamma;D)\,d\beta.
\) \Cref{thm:main-grid-feasibility} then implies that for all sufficiently large \(n\), with probability at least
\(1-\delta_{n,Q}\) over the sampled instance,
\begin{equation}
\sup_{\substack{\beta\in\mathbb R\\
                \gamma\in\Gamma_Q}}
P_\Pi^{(\mathrm{gen})}(\beta,\gamma;D)
\le u_n.
\label{eq:simultaneous-grid-obstruction}
\end{equation}
The event in \cref{eq:simultaneous-grid-obstruction} is uniform
over all parameter-selection procedures, including procedures that
depend on the instance and on previous measurement outcomes. The formal statement is the following.

\begin{corollary}[Adaptive finite-grid sampling obstruction]
\label[corollary]{cor:adaptive-grid-sampling}
Consider any procedure that produces at most \(S\) computational-basis
samples from single-layer generic-QAOA circuits. At sample \(t\),
the procedure may select
\[
\beta_t\in\mathbb R,
\qquad
\gamma_t\in\Gamma_Q
\]
using the full instance description, classical randomness, and all
previous measurement outcomes.

Then
\begin{equation}
\Pr_{D,\mathrm{procedure}}
\left[
\text{at least one raw sample belongs to }\Pi
\right]
\le
\delta_{n,Q}+S u_n.
\label{eq:adaptive-total-sampling-bound}
\end{equation}
In particular, on every instance satisfying
\cref{eq:simultaneous-grid-obstruction},
\begin{equation}
\Pr_{\mathrm{procedure}}
\left[
\text{at least one feasible sample}
\,\middle|\,D
\right]
\le S u_n.
\label{eq:adaptive-fixed-instance-bound}
\end{equation}
\end{corollary}

\begin{proof}
Fix an instance \(D\) satisfying
\cref{eq:simultaneous-grid-obstruction}.
Conditional on any history before sample \(t\), the selected
parameters remain admissible. Hence the conditional probability
that this sample is feasible is at most \(u_n\).

Let \(X_t\) be the indicator that sample \(t\) is feasible.
Conditional expectation gives
\[
\mathbb E[X_t\mid D]\le u_n.
\]
Therefore
\[
\Pr\left[
\sum_{t=1}^S X_t\ge1
\,\middle|\,D
\right]
\le
\mathbb E\left[
\sum_{t=1}^S X_t
\,\middle|\,D
\right]
\le S u_n.
\]
No independence between the selected parameter pairs is required. Finally, the probability of an instance outside the simultaneous
good event is at most \(\delta_{n,Q}\). Splitting according to this
event proves \cref{eq:adaptive-total-sampling-bound}.
\end{proof}

\paragraph{Algorithmic significance.}

Also, it is worth emphasizing that polynomial grid cardinality is sufficient but not necessary as evident from \cref{eq:main-admissible-grid-size}. Thus, the exponent in the union bound identifies the competition between
the number of available cost angles and the suppression at each
fixed angle:
\[
\delta_{n,Q}
\le
\exp\left[
\log Q-\frac{c_*}{2}n^2
\right].
\]
When the right-hand side ceases to be small, the estimate becomes
inconclusive. The adaptive sampling argument of
\Cref{cor:adaptive-grid-sampling} implies that any
procedure producing at most $S$ such samples satisfies
\begin{equation}
    \Pr[\text{at least one feasible sample}]
    \le
    \delta_{n,Q}+S u_n.
    \label{eq:main-adaptive-sampling-bound}
\end{equation}
The budget $S$ includes samples used during training
and after parameter selection.

The budget \(S\) may include samples used to estimate training
objectives or gradients, samples used to select among candidate
parameters, and samples generated after training, provided the
circuits counted in the statement use the admissible parameter
set. Additional classical information about the instance does not
invalidate the bound. For a target success probability \(r\in(0,1)\), any such procedure
with unconditional success probability at least \(r\) must satisfy
\begin{equation}
S
\ge
\frac{r-\delta_{n,Q}}{u_n},
\qquad
r>\delta_{n,Q}.
\label{eq:generic-sample-lower-bound}
\end{equation}
On a good instance, the corresponding necessary condition is
\(S\ge r/u_n\). Thus the number of circuit samples required for constant probability
of obtaining even one feasible output is exponential in \(n^2\),
within the stated finite-grid model. Numerical manifestation of this is present in \Cref{fig:wi4_wi5_grid}. Figures~\ref{fig:wi4_counts} and~\ref{fig:wi5_counts} show the depth-one feasibility counts in our numerical simulation of TSP problems on 4 and 5 locations taken from the benchmark sets in Ref\cite{Osaba2024Qoptlib}. \Cref{fig:Blockqaoa_feasibility} summarizes the simulated
feasibility-mass comparison.

\subsection{Alternative Lenses}
\label{sec:alternate}

The harmonic analysis developed above is complemented by three
further approaches to feasibility concentration outlined in detail in Appendix \ref{app:lenses}. Parameter
averaging, developed in Appendix~\ref{app:avg-angles}, identifies
the interference that survives averaging over the cost angle.
Fourth-moment estimates, derived in Appendix~\ref{app:L4},
measure the concentration of the output distribution.
Light-cone analysis, developed in Appendix \ref{app:lightcone}, describes the support
and overlap of the evolved constraints.
These approaches examine the same probability
\[
    P_\Pi(\beta,\gamma)
    =\sum_{x\in\Pi}|a_{\beta,\gamma}(x)|^2
\]
through its phase structure, amplitude moments, and correlation
geometry. Their common purpose is to identify how the objective and mixer
coordinate probability across the configurations required by
permutation feasibility. The harmonic decomposition resolves this
question by spectral degree. The three additional approaches
resolve it through energy fibres, energy-sum resonances, and
constraint supports. Each yields a quantitative statement about
one component of the concentration mechanism. Together, they
explain the obstruction and identify the structures that a
constraint-enhanced ansatz can exploit. We summarize the essence of each alternative lens below and refer to \cref{app:lenses} for the full detail.

\paragraph{Parameter averaging and energy fibres.}

For a fixed instance, suppose the cost values lie on a
common lattice,
\(C(y)=\omega h(y)\), where \(\omega>0\) is the lattice
spacing and \(h(y)\in\mathbb Z\) is the integer cost
label of configuration \(y\). Define the cost-phase
factor by
\(\phi_\gamma(y)=e^{-i\gamma C(y)}\).
Averaging the cost angle \(\gamma\) uniformly over
\([0,2\pi/\omega]\) gives
\begin{equation}
    \mathbb E_\gamma
    \left[
        \phi_\gamma(y)\overline{\phi_\gamma(z)}
    \right]
    =
    \mathbf 1_{\{C(y)=C(z)\}},
    \label{eq:results-angle-average}
\end{equation}
where the overline denotes complex conjugation,
\(\mathbb E_\gamma\) denotes the uniform angle average,
and the indicator on the right equals one when
\(y\) and \(z\) have the same cost and zero otherwise.
Thus the average retains interference between
equal-cost configurations and cancels interference
between configurations with different costs.
The surviving interference is therefore organized by the
energy fibres
\[
    \mathcal L_E=\{y\mid C(y)=E\}.
\]
Writing
$U_{xy}(\beta)=\langle x|e^{-i\beta\sum_jX_j}|y\rangle$,
\Cref{lem:avg-gamma-lattice} gives
\[
    \mu_\Pi(\beta)
    =\mathbb E_\gamma P_\Pi(\beta,\gamma)
    =2^{-N}\sum_{x\in\Pi}\sum_E
       \left|
          \sum_{y\in\mathcal L_E}U_{xy}(\beta)
       \right|^2.
\]
This identity converts phase averaging into a calculation of
coherent mixer contributions within equal-energy sets.
For a nondegenerate spectrum it gives the uniform mean
$|\Pi|2^{-N}$. More generally, the largest fibre size
$g_C=\max_E|\mathcal L_E|$ gives
\[
    \mu_\Pi(\beta)\le\min\{1,g_C|\Pi|2^{-N}\}.
\]
\Cref{prop:markov} turns this mean into a
typical-angle bound,
\[
    \Pr_\gamma\!\left[
        P_\Pi(\beta,\gamma)>t\mu_\Pi(\beta)
    \right]\le t^{-1},
    \qquad t>1.
\]

The physical content is that cost-angle averaging selects
equal-energy coherence. The size and internal mixer structure
of the energy fibres determine how much feasible probability
survives. Appendix~\ref{app:avg-angles} develops this viewpoint in detail,
including discrete-grid averaging in
\Cref{rem:discrete-grid} and widening-window averages
for real spectra in \Cref{lem:RL-averaging}.

\paragraph{Fourth moments and energy-sum resonances.}

The fourth moment analysis is useful because it measures concentration across the output
distribution. Cauchy--Schwarz gives the meta bound
\begin{equation}
    P_\Pi(\beta,\gamma)
    \le
    \sqrt{|\Pi|}
    \left(\sum_x|a_{\beta,\gamma}(x)|^4\right)^{1/2}.
    \label{eq:CS-L4m}
\end{equation}
Under the lattice cost-angle average of
\cref{eq:results-angle-average}, group ordered
pairs of computational-basis configurations by
their total cost. For each attained value \(E\), define
\[
    \mathcal R_E
    =
    \left\{
        (y,z)\in\{0,1\}^N\times\{0,1\}^N
        \,\middle|\,
        C(y)+C(z)=E
    \right\}.
\]
Thus \(\mathcal R_E\) contains all ordered pairs
with total cost \(E\). In the fourth-moment
expansion, angle averaging retains terms satisfying
\(C(y)+C(z)=C(y')+C(z')\), so contributions within
each \(\mathcal R_E\) remain coherent.

Writing
\(
    U_{xy}(\beta)
    =
    \langle x|e^{-i\beta H_M}|y\rangle
\)
for the mixer transition amplitude from \(y\) to \(x\),
\Cref{lem:L4-avg-gamma-correct} expresses the averaged fourth moment as
\begin{equation}
    M_4(\beta)
    :=
    \mathbb E_\gamma
    \left[
        \sum_x |a_{\beta,\gamma}(x)|^4
    \right]
    =
    2^{-2N}\sum_x\sum_E
    \left|
        \sum_{(y,z)\in\mathcal R_E}
        U_{xy}(\beta)U_{xz}(\beta)
    \right|^2,
    \label{eq:results-L4}
\end{equation}
where \(x\) ranges over all computational-basis
configurations and \(E\) ranges over the distinct
attained sums of two cost values. The factor
\(2^{-2N}\) comes from the uniform initialization.

\Cref{eq:results-L4} shows how the fourth moment probes
resonances between pairs of energies. Their coherent weights determine the concentration
estimate
\begin{equation}
    \mathbb E_\gamma P_\Pi(\beta,\gamma)
    \le
    \min\{1,\sqrt{|\Pi|M_4(\beta)}\}.
    \label{eq:results-L4-feasible}
\end{equation}
The feasible-mass bound in \cref{eq:results-L4-feasible} and
\Cref{prop:typical-L4} supply the corresponding
typical-angle guarantee. The product structure of the mixer also yields
\[
    \sum_x|a_{\beta,\gamma}(x)|^4
    \le
    \left(
        \frac{1+|\sin(2\beta)|}{2}
    \right)^N
\]
for every cost angle. For $N=n^2$ and $|\Pi|=n!$,
\cref{eq:CS-L4m} therefore gives
$e^{-\Omega(n^2)}$ feasibility suppression on mixer-angle
domains satisfying $|\sin(2\beta)|\le1-\varepsilon$ for
fixed $\varepsilon>0$. The resonance formula retains the
objective dependence throughout the full mixer-angle range.
Appendix~\ref{app:L4} derives these identities and bounds.

\paragraph{Light-cone support and constraint overlap.}

Light-cone analysis examines the coordination required by the
constraints in the interaction geometry of the circuit.
Let $P_i$ project onto the one-hot sector of row $i$.
Its support comprises the entire row, and its evolution obeys
\[
    \operatorname{supp}(U_p^\dagger P_iU_p)
    \subseteq
    \mathcal N_p(\operatorname{Row}_i),
\]
where the neighborhood is taken in the full cost-interaction
hypergraph. These enlarged supports identify which constraints
can share correlations during the evolution. For the TSP encoding, the quadratic column penalties connect
every pair of rows within each column. The neighborhood of a complete row therefore reaches all $n^2$ qubits after one cost layer. The transition-fibre estimate supplies quantitative
control within this globally overlapping constraint block. The support derivations appear in
Eqs.~\eqref{eq:single-layer-support}--%
\eqref{eq:constraint-lightcone-p}.
\Cref{lem:finner-light-cone} controls classical collision
events through their dependence on independent coordinates,
and \Cref{lem:block-lightcone-bound} organizes quantum
constraint blocks using disjoint enlarged light cones and
H\"older's inequality. \Cref{lem:fubini-product} justifies
the exchanges of instance averages and angle integrals used
in combining these estimates.

\medskip
The three viewpoints connect spectral interference to
probability concentration and constraint geometry. They also illuminate the CE--QAOA construction, which places row validity
in the encoded state space and preserves it under mixing.
The higher-depth analysis that follows combines the
transition-feature structure with linked cost-angle schedules
and mixer alignment to control repeated applications of the
same objective.

\subsection{Mixer optimization and the linear-depth window}
\label{sec:finitep}

We extend the depth-one feasibility bound to a family of
linked cost-angle schedules at higher depth. The same
sampled edge-cost instance is used at every layer, while
the mixer angles remain independently variable.

Let
\[
N=n^2,\qquad
m=\lfloor n/2\rfloor,\qquad
\ell=n-m,\qquad
d=m\ell,
\]
where $m$ and $\ell$ specify the balanced partition used
in the transition-feature estimates  derived in
App.~\ref{app:lightcone}. For ATSP, set $s=0$.
For symmetric TSP, define
\[
s=\operatorname{nullity}
(\mathsf S+\mathsf S^{\mathsf T})\le2,
\]
where $\mathsf S$ is the cyclic shift matrix.

Fix a maximum depth $P\ge1$ and set $b_n=4n+1$.
Choose a prime $Q>b_n^P$. We construct the linked
cost-angle schedules using a single parameter
$t\in\mathbb Z_Q$:
\begin{equation}
a_j(t)=t b_n^{j-1}\pmod Q,
\qquad
\gamma_j(t)=\frac{2\pi a_j(t)}{Q},
\qquad 1\le j\le P.
\label{eq:linked-cost-angle-schedule}
\end{equation}
where residues are taken in $\{0,\ldots,Q-1\}$.
At depth $p\le P$, the cost-angle schedule consists
of its first $p$ entries, while the $p$ mixer angles
may be chosen independently. The edge-cost ensemble
and deterministic penalties are those of
\Cref{sec:formulation-tsp-ensemble}.

The higher-depth analysis combines the support bounds
\eqref{eq:single-layer-support}--
\eqref{eq:constraint-lightcone-p}, derived in
App.~\ref{app:lightcone}, with the constraint-block
estimate of \Cref{lem:block-lightcone-bound}.
These bounds describe the propagation and overlap
of the evolved constraints. Finner--H\"older
(\Cref{lem:finner-light-cone}) controls classical
collision events on independent coordinates, while
Fubini--Tonelli (\Cref{lem:fubini-product}) justifies
the exchanges of instance averages, collision sums,
and angle integrals. Together, these estimates
support the averaged evolution under linked schedules
established in \Cref{prop:linked-averaged-channel}. To extend the resulting ensemble estimate uniformly
over all mixer angles, we use multivariate Fourier
orthogonality.

\begin{lemma}[Depth-$p$ mixer alignment]
\label[lemma]{lem:depth-p-mixer-alignment}
For every fixed instance and cost schedule,
\begin{equation}
\sup_{\boldsymbol\beta}
P_{\Pi,p}^{(\mathrm{gen})}
(\boldsymbol\beta,\boldsymbol\gamma;D)
\le
(N+1)^p
\overline P_{\Pi,p}(\boldsymbol\gamma;D),
\label{eq:depth-p-alignment}
\end{equation}
where
\[
\overline P_{\Pi,p}(\boldsymbol\gamma;D)
=
\frac1{\pi^p}
\int_{[0,\pi]^p}
P_{\Pi,p}^{(\mathrm{gen})}
(\boldsymbol\beta,\boldsymbol\gamma;D)
\,d\boldsymbol\beta.
\]
\end{lemma}

\begin{proof}
The transverse-field mixer has eigenvalues
$N-2k$, $k=0,\ldots,N$. Expanding the feasible
output component in the corresponding mixer
frequencies gives, up to an overall phase,
\[
\widehat\Pi|\psi_p\rangle
=
\sum_{\boldsymbol k\in\{0,\ldots,N\}^p}
e^{2i\boldsymbol k\cdot\boldsymbol\beta}
|v_{\boldsymbol k}\rangle,
\]
where $|v_{\boldsymbol k}\rangle$ depends on the
instance and cost-angle schedule but is independent
of $\boldsymbol\beta$.

By Cauchy--Schwarz,
\[
P_{\Pi,p}^{(\mathrm{gen})}
(\boldsymbol\beta,\boldsymbol\gamma;D)
\le
(N+1)^p
\sum_{\boldsymbol k}
\|v_{\boldsymbol k}\|^2.
\]
Multivariate Fourier orthogonality gives
\[
\frac1{\pi^p}
\int_{[0,\pi]^p}
e^{2i(\boldsymbol k-\boldsymbol l)
\cdot\boldsymbol\beta}
\,d\boldsymbol\beta
=
\delta_{\boldsymbol k,\boldsymbol l}.
\]
Integrating the squared norm and applying
\Cref{lem:fubini-product} therefore yields
\[
\overline P_{\Pi,p}(\boldsymbol\gamma;D)
=
\sum_{\boldsymbol k}
\|v_{\boldsymbol k}\|^2.
\]
Substitution establishes
\cref{eq:depth-p-alignment}.
\end{proof}

We now combine the mixer-alignment with
the ensemble estimate for linked cost-angle
schedules to obtain a simultaneous feasibility
bound across depths and parameter choices.

\begin{theorem}[Uniform feasibility bound over depths
and linked schedules]
\label[theorem]{thm:linked-higher-depth}
For the ensemble and schedules defined above, let
\[
u_{n,p}=(N+1)^p n!\,2^{-d/2},
\]
and
\[
\delta_{n,P,Q}
=
\min\left\{
1,\,
P(Q-1)2^{-d/2+m+s\ell}
\right\}.
\]
With probability at least $1-\delta_{n,P,Q}$
over the sampled instance, the bound
\begin{equation}
\sup_{\substack{t\in\mathbb Z_Q\\
\boldsymbol\beta\in\mathbb R^p}}
P_{\Pi,p}^{(\mathrm{gen})}
(\boldsymbol\beta,\boldsymbol\gamma(t);D)
\le u_{n,p}
\label{eq:linked-uniform-depth-bound}
\end{equation}
holds simultaneously for every $1\le p\le P$.
\end{theorem}

\begin{proof}
For every nonzero $t\in\mathbb Z_Q$,
\Cref{prop:linked-averaged-channel}, together
with \Cref{lem:fubini-product}, gives
\[
\mathbb E_D
\overline P_{\Pi,p}(\boldsymbol\gamma(t);D)
\le
n!2^{-d+m+s\ell}.
\]
Applying Markov's inequality at the threshold
$n!2^{-d/2}$ yields
\[
\Pr_D\left[
\overline P_{\Pi,p}(\boldsymbol\gamma(t);D)
>
n!2^{-d/2}
\right]
\le
2^{-d/2+m+s\ell}.
\]
By \Cref{lem:depth-p-mixer-alignment}, whenever
the averaged feasible probability is below this
threshold, its supremum over all mixer angles
is bounded by
\[
(N+1)^p n!2^{-d/2}=u_{n,p}.
\]
A union bound over the $P(Q-1)$ nonzero
schedule--depth pairs therefore gives the
claimed failure probability.

For $t=0$, every cost angle in
\cref{eq:linked-cost-angle-schedule} vanishes.
The transverse-field mixer preserves the uniform
initial state up to a global phase, so
\[
P_{\Pi,p}^{(\mathrm{gen})}
(\boldsymbol\beta,\boldsymbol\gamma(0);D)
=
n!2^{-N}\le u_{n,p}
\]
for every mixer-angle choice. Thus the same bound
holds for all linked schedules, completing the proof.
\end{proof}

The uniform estimate yields an exponential separation
between generic QAOA and CE--QAOA throughout a
linear-depth window.

\begin{corollary}[Exponential separation throughout
a linear-depth window]
\label[corollary]{cor:linked-linear-depth-separation}
Fix $\alpha>0$, let $P=\lfloor\alpha n\rfloor$,
and choose a prime satisfying
\[
(4n+1)^P<Q<2(4n+1)^P.
\]
For all sufficiently large $n$, with probability
$1-e^{-\Omega(n^2)}$ over this ensemble,
\[
\sup_{t,\boldsymbol\beta}
P_{\Pi,p}^{(\mathrm{gen})}
(\boldsymbol\beta,\boldsymbol\gamma(t);D)
\le
\exp\left[
-\frac{\log2}{8}n^2+O(n\log n)
\right]
\]
simultaneously at every depth $1\le p\le P$.

Define the optimized CE--QAOA feasible probability
over the same parameter family by
\[
P_{p,\star}^{(\mathrm{CE})}(D)
=
\sup_{\substack{t\in\mathbb Z_Q\\
\boldsymbol\beta\in\mathbb R^p}}
P_{\Pi,p}^{(\mathrm{CE})}
(\boldsymbol\beta,\boldsymbol\gamma(t);D).
\]
Then
\[
P_{p,\star}^{(\mathrm{CE})}(D)
\ge \frac{n!}{n^n},
\]
and, on the same event,
\begin{equation}
\frac{
P_{p,\star}^{(\mathrm{CE})}(D)
}{
\displaystyle\sup_{t,\boldsymbol\beta}
P_{\Pi,p}^{(\mathrm{gen})}
(\boldsymbol\beta,\boldsymbol\gamma(t);D)
}
\ge
\frac{2^{d/2}}{(N+1)^p n^n}
=
e^{\Omega(n^2)}
\label{eq:linked-linear-depth-ratio}
\end{equation}
uniformly throughout the depth window.
\end{corollary}

\begin{proof}
Apply \Cref{thm:linked-higher-depth} with
$P=\lfloor\alpha n\rfloor$. The chosen prime
satisfies
\[
\log Q
=
P\log(4n+1)+O(1)
=
O(n\log n).
\]
Since
\[
d=\frac{n^2}{4}+O(1),
\qquad
m+s\ell=O(n),
\]
the failure probability obeys
\[
\log\!\left(
P(Q-1)2^{-d/2+m+s\ell}
\right)
=
-\frac{\log2}{8}n^2+O(n\log n).
\]
Consequently, the simultaneous bound holds
with probability $1-e^{-\Omega(n^2)}$.

For the generic feasible-probability envelope,
\[
\begin{aligned}
\log u_{n,p}
&=
-\frac d2\log2
+p\log(N+1)+\log(n!)\\
&=
-\frac{\log2}{8}n^2
+O(n\log n),
\end{aligned}
\]
uniformly for $p\le P$. This establishes the
generic-QAOA bound.

For CE--QAOA, set all mixer angles to zero.
The diagonal cost layers then preserve the
uniform measurement probabilities on the
$n^n$ configurations of the product one-hot
space. Since $n!$ of these configurations are
feasible, this gives
\[
P_{p,\star}^{(\mathrm{CE})}(D)
\ge \frac{n!}{n^n}
\]
at every depth and for every linked cost-angle
schedule, realizing the baseline in
\cref{eq:two-feasibility-baselines}.

Comparing this attainable CE--QAOA baseline
with the generic upper bound $u_{n,p}$ gives
\cref{eq:linked-linear-depth-ratio}. Indeed,
\[
\begin{aligned}
\log\left(
\frac{2^{d/2}}{(N+1)^p n^n}
\right)
&=
\frac d2\log2
-p\log(N+1)-n\log n
=
\frac{\log2}{8}n^2-O(n\log n).
\end{aligned}
\]
The leading quadratic term is positive and
uniform throughout $p\le P$, establishing
the claimed exponential separation.
\end{proof}

% The transition-fibre estimate supplies the
% quantitative bound for the globally overlapping
% TSP block. Linked schedules propagate this
% estimate through repeated uses of the same
% objective, and mixer alignment controls all
% mixer angles.

\section{Discussion}
\paragraph{Adaptive sampling.} Note that at linear depth, the family contains
\[
Q=\exp[\Theta(n\log n)]
\]
linked cost-angle schedules. The schedule
count and mixer-alignment factor contribute
$O(n\log n)$ to the exponents, leaving the
quadratic suppression in
\Cref{cor:linked-linear-depth-separation}. On the simultaneous good-instance event,
every admissible circuit of depth at most
$P$ has feasible probability at most $u_{n,P}$.
An adaptive procedure generating at most
$S$ raw circuit samples therefore satisfies
\[
\Pr[\text{at least one feasible sample}]
\le
\delta_{n,P,Q}+S u_{n,P}.
\]
The procedure may use the instance and its
previous measurement outcomes to select
the next depth and parameters within the
linked family. The sample budget includes measurements used
during parameter search and samples generated
after training. Polynomial sample budgets give
success probability $e^{-\Omega(n^2)}$.
Conversely, constant-probability feasible
sampling requires $S=e^{\Omega(n^2)}$
throughout the stated linear-depth regime.

\paragraph{Relation to known constant-depth limitations and OGP barriers.}
Our analysis connects with several established results on
shallow depth quantum optimization circuits.
For circuits with geometrically local gates, or bounded-depth
circuits on bounded-degree interaction graphs, light cones
restrict the propagation of correlations
\cite{Hastings2010Locality}. For QAOA on sparse graphs,
this limits the portion of the instance accessible to local
observables at a given depth
\cite{FarhiGamarnikGutmann2020WholeGraph}.
Related results characterize the computational power of
constant-depth quantum circuits
\cite{BravyiGossetKonig2018}.
For certain random optimization ensembles, the solution
space exhibits the \emph{Overlap Gap Property} (OGP),
which yields obstructions for algorithmic classes and
depth regimes covered by the corresponding results
\cite{GamarnikZadik2022QAOAOGP}. Further studies examine variational landscape structure
\cite{Anschuetz_2024}, fixed-depth QAOA performance
\cite{Basso2022}, and symmetry-dependent limitations
of shallow quantum algorithms
\cite{Bravyi_2020}. In the TSP encoding considered here,
the structural light cones of the row constraints overlap
globally after one cost layer. The transition-feature
estimates provide the quantitative feasibility bounds
in this globally overlapping geometry.

\paragraph{Constraint support and light-cone overlap.}
Let $\operatorname{Row}_i=\{(i,a):a\in[n]\}$ denote the
qubits encoding row $i$, and let $P_i$ project onto its
one-hot sector. Since $P_i$ acts on the entire row,
its Heisenberg evolution satisfies
\[
    \operatorname{supp}(U_p^\dagger P_iU_p)
    \subseteq
    \mathcal N_p(\operatorname{Row}_i),
\]
where $\mathcal N_p$ is the radius-$p$ neighborhood in
the full cost-interaction hypergraph. Both intra-row
and inter-row interactions contribute to this
neighborhood. The quadratic column penalties couple every pair of
rows within each column. Consequently, the neighborhood of a complete row already contains all $n^2$ qubits after one cost layer. This support geometry is included
in the constraint-block analysis of
\Cref{lem:block-lightcone-bound}. That lemma combines constraints whose enlarged backward
light cones are disjoint and uses H\"older's inequality
to account for overlap between groups. Its 
input is a bound on the reduced state associated with
each constraint block. For the TSP ensemble, the
transition-fibre estimate supplies this input for the
full permutation block. \Cref{cor:linked-linear-depth-separation}
shows that the transition-feature structure sustains
exponential feasibility suppression with high probability  throughout
$p\le\alpha n$, even when the  row light
cones overlap globally.

% ------------------------------------------------------------
%  Block- vs Generic-QAOA: success-probability gap (p = 1)

\section{Conclusion and Outlook}
\label{sec:conclusion}

This work characterizes the performance of generic QAOA
through feasibility concentration. For random edge-cost
TSP/ATSP ensembles encoded on $n^2$ qubits, we establish
exponentially small feasible probability with high
probability over the instance. At depth one, the bound
holds simultaneously over all mixer angles and a
cost-angle grid whose cardinality can grow exponentially
with $n^2$. It therefore controls the feasible probability
attained after parameter optimization within that domain.

The analysis combines Boolean harmonic structure,
objective-phase averaging, and transition-feature
collision estimates. Light-cone support bounds describe
the overlap of the evolved constraints, and
transition-fibre estimates control the averaged state.
Linked cost-angle schedules propagate this control
through repeated uses of the same objective.
The resulting feasibility bound holds simultaneously
throughout $p\le\alpha n$, for fixed $\alpha>0$,
over the linked schedules and all mixer angles.

These uniform bounds also control adaptive parameter
selection. An algorithm may choose its next admissible
circuit using the instance and previous measurement
outcomes but the total number of raw circuit samples
required for constant-probability feasible sampling
remains exponential in $n^2$. This budget includes
measurements used during training and samples generated
after parameter selection. The results thus connect
the feasibility geometry and objective structure
to the sampling cost of parameter search.

The CE--QAOA kernel supplies a route to exponential
enhancement through constraint embedding. Its product
one-hot initialization and block-local $XY$ mixer
preserve the encoded space. At every depth, setting
the mixer angles to zero attains feasible probability
$n!/n^n$ under the diagonal objective evolution.
Comparing this attainable baseline with the generic
upper bounds gives an $e^{\Omega(n^2)}$ enhancement
throughout the proved depth window. This demonstrates
how embedding part of the constraint structure into
the initialization and mixer can improve the
probability of generating valid optimization candidates.

The three-band Walsh decomposition identifies a further
question. The low- and high-degree spectral contributions
to feasible probability are exponentially suppressed by
the geometry of the permutation sector, whereas the middle
band admits no corresponding geometry-only exponential
bound. Consequently, any feasible-probability enhancement
beyond the spectral-edge contributions must arise
from the middle band. Understanding how objective phases and 
mixer dynamics produce the middle-band contribution, and how 
this relates to coherent amplitude transport, is therefore 
a natural direction for further investigation.

The framework also motivates further study of the
interaction between parameter resolution, objective
correlations, and feasible probability. Extending the
ensemble estimates to broader cost-angle families would
enlarge the parameter domains covered by the sampling
guarantees. Within the CE--QAOA space, harmonic and
mixer-walk analyses can investigate how objective phases
and block mixing transport amplitude toward feasible
configurations and amplify feasible mass beyond the
initialization baseline. These open points provide a starting point
for a subsequent study of coherent amplitude transport
and its utilization in constrained quantum optimization.

\appendix

\section{Materials and Methods}
\label{sec:methods}

\subsection{Krawtchouk Polynomials on the Hamming sphere}
\label{sec:kraw}
We continue the Walsh analysis on permutation encoding from
\cref{sec:walsh} to develop the Krawtchouk representation of Hamming-sphere
indicators for spectral estimates\cite{feinsilver2007,ODonnell2014}.

\medskip
\noindent
Let $g_w:\{0,1\}^n\to\{0,1\}$ denote the Hamming–sphere indicator $g_w(x)=\mathbf{1}_{\{|x|=w\}}$.
In particular, the one–hot indicator is $g_1$ (support $\{x:|x|=1\}$). For $0\le w\le n$, the (unnormalized) binary Krawtchouk polynomials are
\begin{equation}
  K_w^{(n)}(r)\ :=\ \sum_{j=0}^{w}(-1)^j \binom{r}{j}\binom{n-r}{\,w-j\,},
  \qquad r=0,1,\dots,n.
  \label{eq:Kraw-def}
\end{equation}
They satisfy the following generating function and orthogonality identities respectively~\cite{feinsilver2007}
\begin{align}
  \sum_{w=0}^{n} K_w^{(n)}(r)\, z^w \ &=\ (1+z)^{\,n-r}(1-z)^{\,r}, \label{eq:Kraw-gen}\\
  \sum_{r=0}^{n} \binom{n}{r}\, K_w^{(n)}(r)\,K_{w'}^{(n)}(r) \ &=\ 2^{n}\binom{n}{w}\,\delta_{w,w'}.
  \label{eq:Kraw-orth}
\end{align}

\begin{lemma}[Krawtchouk transform of a Hamming sphere]
\label[lemma]{lem:layer-sum-kraw}
For $S\subseteq[n]$ with $|S|=r$,
\[
  \sum_{x:\,|x|=w} \chi_S(x) \;=\; K_w^{(n)}(r).
\]
Consequently,
\begin{equation}
  \widehat{g_w}(S) \;=\; 2^{-n}\,K_w^{(n)}\!\bigl(|S|\bigr).
  \label{eq:gw-hat}
\end{equation}
\end{lemma}

\begin{proof}
Partition the coordinates into the $r$ positions where $S$ has 1s and the $n-r$ positions where $S$ has 0s.
To form an $x$ with $|x|=w$, choose $j$ ones among the $r$ ``active'' positions and $w-j$ ones among the remaining $n-r$ positions.
There are $\binom{r}{j}\binom{n-r}{w-j}$ such $x$, and each contributes a phase $(-1)^j$ because $\chi_S(x)=(-1)^{S\cdot x}$ counts parity on the $r$ active positions only.
Summing over $j$ gives \eqref{eq:Kraw-def}, hence the first claim.
Applying the Walsh transform definition to $g_w$,
\[
  \widehat{g_w}(S)=2^{-n}\sum_{x} g_w(x)\chi_S(x) 
  = 2^{-n}\sum_{|x|=w}\chi_S(x)=2^{-n}K_w^{(n)}(|S|).
\]
\end{proof}

\[
  K_{1}^{(n)}(r)\ =\ \binom{r}{0}\binom{n-r}{1}-\binom{r}{1}\binom{n-r}{0}
  \ =\ (n-r)-r\ =\ n-2r.
\]
Thus, for $g_1$ (the ``exactly one'' indicator),
\begin{equation}
  \widehat{g_1}(S)\ =\ 2^{-n}\bigl(n-2|S|\bigr),
  \qquad S\subseteq[n],
  \label{eq:onehot-transform}
\end{equation}
which matches the direct computation $\widehat{g_1}(S)=2^{-n}\sum_{i=1}^{n}(-1)^{\mathbf{1}_{\{i\in S\}}}=2^{-n}(n-2|S|)$.

\begin{lemma}[Low-degree Fourier weight for $g_1$]
\label[lemma]{lem:onehot-lowdeg}
For $n\ge2$ and an integer cutoff $d\ge0$,
\[
 \sum_{|S|\le d}|\widehat{g_1}(S)|^2
 =2^{-2n}\sum_{r=0}^{\min(d,n)}\binom nr(n-2r)^2
 \le(d+1)n^{d+2}2^{-2n}.
\]
The fraction of the total Fourier energy $n2^{-n}$ in these degrees
is at most $(d+1)n^{d+1}2^{-n}$.
\end{lemma}

\begin{proof}
\Cref{eq:onehot-transform} gives the equality. For
$0\le r\le\min(d,n)$, use $\binom nr\le n^r$ and $(n-2r)^2\le n^2$ to obtain
\[
 \sum_{r=0}^{\min(d,n)}\binom nr(n-2r)^2
 \le n^2\sum_{r=0}^{\min(d,n)}n^r
 \le(d+1)n^{d+2}.
\]
Parseval with the normalization of \cref{eq:ft} gives the total
Fourier energy $2^{-n}\|g_1\|_2^2=n2^{-n}$.
\end{proof}

The low-degree estimate extends to permutation feasibility through
tensorization within the disjoint rows and within the disjoint columns,
followed by Fourier convolution of their indicators.

\subsection{Harmonic analysis of the probability profile}
\label{sec:route1}

The following auxiliary estimates control the factorial factor and normalization in the permutation-indicator bound. The factorial inequality appears 
in the proof of Lemma~2 on page~64 of Alon and Spencer
\cite{AlonSpencer2008}. We reproduce its proof here
for completeness. 

\begin{lemma}
\label[lemma]{lem:ratio-nn-over-nfact}
For every integer $n\ge 2$,
\[
\bigl((n-1)!\bigr)^{\,n(n+1)}\cdot \bigl((n+1)!\bigr)^{\,n(n-1)}
\;<\;
\bigl(n!\bigr)^{\,2(n^2-1)},
\]
equivalently,
\[
\left(\frac{n^n}{n!}\right)^{\!2}
\;>\;
\left(\frac{n+1}{n}\right)^{\,n(n-1)}.
\]
\end{lemma}

\begin{proof}
The two forms are equivalent by dividing both sides of the first inequality by
$\bigl((n-1)!\,(n+1)!\bigr)^{\,n(n-1)}$ and simplifying.

\emph{Base case $n=2$.} We have
\[
\left(\frac{2^2}{2!}\right)^{\!2}= \left(\frac{4}{2}\right)^2=4
\quad\text{and}\quad
\left(\frac{3}{2}\right)^{2}= \frac{9}{4},
\]
so the inequality holds.

\emph{Case $n\ge 3$.} By Arithmetic Mean–Geometric Mean inequality ( AM--GM),
\[
n!\;=\;\prod_{k=1}^n k \;\le\; \left(\frac{1+\cdots+n}{n}\right)^{\!n}
=\left(\frac{n+1}{2}\right)^{\!n}.
\]
Hence
\[
\left(\frac{n^n}{n!}\right)^{\!2}
\;\ge\;
\left(\frac{n^n}{\bigl(\tfrac{n+1}{2}\bigr)^n}\right)^{\!2}
=
\left(\frac{2n}{n+1}\right)^{\!2n}.
\]
Using $(1+\tfrac{1}{n})^n<e$ we get
\[
\frac{2n}{n+1}
=\frac{2}{1+\tfrac{1}{n}}
> 2e^{-1/n},
\qquad\Rightarrow\qquad
\left(\frac{2n}{n+1}\right)^{\!2n} > 4^n e^{-2}.
\]
For $n\ge 3$ we have $4^n>e^{\,n+1}$, hence $4^n e^{-2}>e^{\,n-1}$. Therefore
\[
\left(\frac{n^n}{n!}\right)^{\!2}
> e^{\,n-1}.
\]
On the other hand, $(1+\tfrac{1}{n})^n<e$ implies
\[
\left(\frac{n+1}{n}\right)^{\,n(n-1)} = \bigl((1+\tfrac{1}{n})^n\bigr)^{\,n-1} < e^{\,n-1}.
\]
Combining the last two displays yields
\[
\left(\frac{n^n}{n!}\right)^{\!2}
\;>\;
\left(\frac{n+1}{n}\right)^{\,n(n-1)}.
\]
% as claimed.
\end{proof}

\begin{lemma}
\label[lemma]{lem:control-ineq}
Let $c_T>0$ and set $a:=\tfrac{c_T\ln 2}{2}$. Then for all
\[
n \;\ge\; n_0(c_T)\ :=\ \max\!\Bigl\{9,\ \Bigl\lceil \tfrac{4}{a^2}\Bigr\rceil\Bigr\}
\]
we have $\,n^n\,2^{-c_T n^2/2}<1$.
\end{lemma}

\begin{proof}
For $n\ge 9$ we have $\ln n<\sqrt{n}$, and for $n\ge 4/a^2$ we have
$\sqrt{n}\le a n$. Hence $\ln n<a n$, so
\(
n\ln n-a n^2 \;<\; 0,
\)
which implies $n^n 2^{-c_T n^2/2}=\exp(n\ln n-a n^2)<1$.
\end{proof}

\begin{lemma}[Low-degree mass of $\mathbf{1}_\Pi$]
\label[lemma]{lem:Pi-lowdeg}
For $n\ge2$ and an integer cutoff $d\ge0$, let
$R_n=(n^n/n!)^2$. Then
\[
 \sum_{|S|\le d}|\widehat{\mathbf{1}_\Pi}(S)|^2
 \le(d+1)n^{2d}2^{-2n^2}|\Pi|^2R_n,
 \qquad
 R_n>\left(\frac{n+1}{n}\right)^{n(n-1)}.
\]
\end{lemma}

\begin{proof}
Use the row and column indicators $R$ and $I_{\mathrm{col}}$ from
\cref{sec:walsh}, with $\mathbf{1}_\Pi=R I_{\mathrm{col}}$.
For a mask $S$, let $S_i$ denote its restriction to row $i$ and
$S^{(j)}$ its restriction to column $j$. \Cref{eq:onehot-transform}
and tensorization across the disjoint row blocks give
\begin{equation}
 \widehat R(S)=2^{-N}\prod_{i=1}^n(n-2|S_i|),
 \qquad \|\widehat R\|_2^2=2^{-N}n^n.
 \label{eq:Rhat-factor}
\end{equation}
The column factorization gives
$\widehat I_{\mathrm{col}}(S)=2^{-N}\prod_{j=1}^n(n-2|S^{(j)}|)$
and $\|\widehat I_{\mathrm{col}}\|_2^2=2^{-N}n^n$.

The Walsh product identity yields
\[
 \widehat{\mathbf{1}_\Pi}(S)
 =\sum_T\widehat R(T)\widehat I_{\mathrm{col}}(S\oplus T).
\]
For $E_d=\{S\subseteq[N]\mid |S|\le d\}$, Cauchy--Schwarz and
shift invariance therefore give
\[
 \sum_{S\in E_d}|\widehat{\mathbf{1}_\Pi}(S)|^2
 \le |E_d|\|\widehat R\|_2^2\|\widehat I_{\mathrm{col}}\|_2^2
 =|E_d|2^{-2N}n^{2n}.
\]
Since
\[
 |E_d|=\sum_{r=0}^{\min(d,N)}\binom Nr\le(d+1)N^d,
\]
we obtain
\begin{equation}
 \sum_{|S|\le d}|\widehat{\mathbf{1}_\Pi}(S)|^2
 \le(d+1)n^{2d}2^{-2n^2}n^{2n}.
 \label{eq:lowdeg-Pi-pre}
\end{equation}
Substitute $n^{2n}=|\Pi|^2R_n$. The lower bound on $R_n$ follows
from \Cref{lem:ratio-nn-over-nfact}.
\end{proof}

Stirling's formula gives $R_n=e^{2n}(2\pi n)^{-1}(1+O(n^{-1}))$.
The next estimate applies the indicator bound to the Walsh truncation
$p^{\le2T}$ of a probability mass function $p$.

\begin{lemma}[Low-degree probability mass contribution]
\label[lemma]{lem:lowdeg-corr-correct}
For every integer cutoff $T\ge0$ and every probability mass function $p$,
\[
 \left|2^{-N}\langle\mathbf{1}_\Pi,p^{\le2T}\rangle\right|
 \le\sqrt{2T+1}\,n^{2T}n^n2^{-3N/2}.
\]
For sufficiently large $n$, this is at most
$\sqrt{2T+1}\,n^{2T}2^{-N}$.
\end{lemma}

\begin{proof}
Parseval with \cref{eq:ft} gives
\begin{equation}
 \sum_S|\widehat p(S)|^2=2^{-N}\sum_xp(x)^2\le2^{-N}.
 \label{eq:pmf-fourier-l2}
\end{equation}
Cauchy--Schwarz and \Cref{lem:Pi-lowdeg}, with $d=2T$, imply
\begin{align*}
 \left|2^{-N}\langle\mathbf{1}_\Pi,p^{\le2T}\rangle\right|
 &\le\left(\sum_{|S|\le2T}|\widehat{\mathbf{1}_\Pi}(S)|^2\right)^{1/2}
       \left(\sum_S|\widehat p(S)|^2\right)^{1/2}\\
 &\le\sqrt{2T+1}\,n^{2T}n^n2^{-3N/2}.
\end{align*}
\Cref{lem:control-ineq} with $c_T=1$ gives the final assertion.
\end{proof}

% ============================================================
\subsection{Deterministic suppression of the spectral edges}
\label{subsec:three-band-edges}
% ============================================================

Fix an integer degree cutoff \(0\le D<N/2\).
Within this subsection, \(D\) denotes the spectral cutoff.
Define
\[
    \mathsf L_D=\{S\subseteq[N]\mid |S|\le D\},
    \qquad
    \mathsf H_D=\{S\subseteq[N]\mid |S|\ge N-D\}.
\]
For \(p_{\beta,\gamma}(x)=|a_{\beta,\gamma}(x)|^2\),
the edge contributions to the Walsh correlation are
\begin{align*}
    P_\Pi^{\mathrm{low}}
    &=
    2^N\sum_{S\in\mathsf L_D}
    \widehat{\mathbf1_\Pi}(S)
    \widehat p_{\beta,\gamma}(S),\\
    P_\Pi^{\mathrm{top}}
    &=
    2^N\sum_{S\in\mathsf H_D}
    \widehat{\mathbf1_\Pi}(S)
    \widehat p_{\beta,\gamma}(S).
\end{align*}
The remaining masks contribute
\(P_\Pi^{\mathrm{mid}}
=P_\Pi-P_\Pi^{\mathrm{low}}-P_\Pi^{\mathrm{top}}\).
These are signed spectral contributions.

Let
\begin{equation}
    B_D
    :=
    |\mathsf L_D|
    =
    \sum_{r=0}^{D}\binom Nr .
    \label{eq:BD-definition}
\end{equation}
The row/column convolution calculation leading to
\cref{eq:lowdeg-Pi-pre} gives, uniformly in \(D\),
\begin{equation}
    \sum_{|S|\le D}
    |\widehat{\mathbf 1_\Pi}(S)|^2
    \le
    B_D\,2^{-2N}n^{2n}.
    \label{eq:three-band-indicator-edge-mass}
\end{equation}

Combining this estimate with \cref{eq:pmf-fourier-l2}, applied
to the output probability mass function $p_{\beta,\gamma}$, gives the
low-frequency contribution. Therefore Cauchy--Schwarz and
\cref{eq:three-band-indicator-edge-mass} imply
\begin{align}
    |P_\Pi^{\mathrm{low}}|
    &\le
    2^N
    \left(
        \sum_{|S|\le D}
        |\widehat{\mathbf 1_\Pi}(S)|^2
    \right)^{1/2}
    \left(
        \sum_S
        |\widehat p_{\beta,\gamma}(S)|^2
    \right)^{1/2}
  \le
    \sqrt{B_D}\,
    2^{-N/2}n^n.
    \label{eq:three-band-low-bound}
\end{align}

The upper spectral edge follows from the fixed Hamming weight of every
permutation matrix.  For \(x\in\Pi\),
\(
    |x|=n,
\)
and hence, using the Walsh group law,
\begin{equation}
    \chi_{S^c}(x)
    =
    \chi_{[N]}(x)\chi_S(x)
    =
    (-1)^n\chi_S(x).
    \label{eq:permutation-complement-character}
\end{equation}
Consequently,
\begin{equation}
    \widehat{\mathbf 1_\Pi}(S^c)
    =
    (-1)^n
    \widehat{\mathbf 1_\Pi}(S).
    \label{eq:permutation-fourier-complement}
\end{equation}
Moreover,
\begin{equation}
    |\widehat p_{\beta,\gamma}(S)|
    =
    2^{-N}
    \left|
        \sum_xp_{\beta,\gamma}(x)\chi_S(x)
    \right|
    \le
    2^{-N}.
    \label{eq:pmf-fourier-pointwise}
\end{equation}
Thus
\begin{align}
    |P_\Pi^{\mathrm{top}}|
    &\le
    \sum_{|S|\ge N-D}
    |\widehat{\mathbf 1_\Pi}(S)|
   =
    \sum_{|T|\le D}
    |\widehat{\mathbf 1_\Pi}(T)|
    \le
    \sqrt{B_D}
    \left(
        \sum_{|T|\le D}
        |\widehat{\mathbf 1_\Pi}(T)|^2
    \right)^{1/2}
    \le
    B_D\,2^{-N}n^n .
    \label{eq:three-band-top-bound}
\end{align}

For
\(
    D=O(\log n),
\)
the standard estimate
\(
    B_D
    \le
    (D+1)N^D
\)
gives
\(
    \log B_D
    =
    O((\log n)^2).
\) Since \(N=n^2\),
\begin{align}
    |P_\Pi^{\mathrm{low}}|
    &\le
    \exp\left[
        -\frac{\ln2}{2}n^2
        +n\ln n
        +O((\ln n)^2)
    \right],
    \label{eq:three-band-low-asymptotic}
    \\
    |P_\Pi^{\mathrm{top}}|
    &\le
    \exp\left[
        -(\ln2)n^2
        +n\ln n
        +O((\ln n)^2)
    \right].
    \label{eq:three-band-top-asymptotic}
\end{align}
Both spectral edges are exponentially suppressed uniformly in the
QAOA angles and independently of the coefficients of $H_C$. The
ensemble analysis in \cref{sec:probabilistic-p1-closure} supplies
the objective-dependent estimate for the full feasible mass.

\subsection{A probabilistic spectral analysis 
for a TSP ensemble}
\label{sec:probabilistic-p1-closure}
We use the phase field, mixer kernel, and Walsh normalization from
\cref{sec:walsh}, together with the TSP/ATSP ensemble specified in
\cref{sec:formulation-tsp-ensemble}. Instance dependence is suppressed
in the notation for amplitudes and probabilities below.
\Cref{lem:walsh-diag} reduces mixer optimization to alignment of
homogeneous Walsh shells. A generating function and a transition-feature
collision estimate then bound their ensemble-averaged feasible mass.

\subsubsection{Walsh-shell alignment}

We begin with a structural statement which is independent of the
particular form of the diagonal cost Hamiltonian.

\begin{lemma}[Single-layer Walsh-shell alignment]
\label[lemma]{lem:single-layer-shell-alignment}
Let
\[
\phi_\gamma(x)
=
\sum_{k=0}^{N}F_k^{(\gamma)}(x)
\]
be the homogeneous Walsh decomposition of the cost-phase field, where
\[
F_k^{(\gamma)}(x)
=
\sum_{\substack{S\subseteq[N]\\|S|=k}}
\widehat{\phi_\gamma}(S)\chi_S(x).
\]
For any target set
$\mathcal A\subseteq\{0,1\}^{N}$ and any
$\mathcal K\subseteq\{0,\ldots,N\}$, define
\[
P_{\mathcal A,\mathcal K}(\beta,\gamma)
=
2^{-N}
\sum_{x\in\mathcal A}
\left|
\sum_{k\in\mathcal K}
e^{2ik\beta}F_k^{(\gamma)}(x)
\right|^2.
\]
Then
\[
\frac{1}{\pi}
\int_0^\pi
P_{\mathcal A,\mathcal K}(\beta,\gamma)\,d\beta
=
2^{-N}
\sum_{x\in\mathcal A}
\sum_{k\in\mathcal K}
|F_k^{(\gamma)}(x)|^2,
\]
and
\[
\sup_{\beta\in\mathbb R}
P_{\mathcal A,\mathcal K}(\beta,\gamma)
\le
|\mathcal K|
\frac{1}{\pi}
\int_0^\pi
P_{\mathcal A,\mathcal K}(\beta,\gamma)\,d\beta.
\]
\end{lemma}

\begin{proof}
\Cref{lem:walsh-diag} gives
\[
 K_\beta\phi_\gamma(x)
 =e^{-iN\beta}\sum_{k=0}^N e^{2ik\beta}F_k^{(\gamma)}(x).
\]
The global phase cancels in the squared modulus.
Expanding the squared modulus and integrating over
$\beta\in[0,\pi]$, orthogonality gives
\[
\frac{1}{\pi}
\int_0^\pi
e^{2i(k-\ell)\beta}\,d\beta
=
\delta_{k\ell},
\]
which proves the first identity. For every $x\in\mathcal A$, Cauchy--Schwarz gives
\[
\left|
\sum_{k\in\mathcal K}
e^{2ik\beta}F_k^{(\gamma)}(x)
\right|^2
\le
|\mathcal K|
\sum_{k\in\mathcal K}
|F_k^{(\gamma)}(x)|^2.
\]
Multiplying by $2^{-N}$ and summing over
$x\in\mathcal A$ proves the second statement.
\end{proof}

Taking $\mathcal A=\Pi$ and
$\mathcal K=\{0,\ldots,N\}$ gives the following immediate consequence.

\begin{corollary}[Full feasible-mass alignment bound]
\label[corollary]{cor:full-feasible-alignment}
For every diagonal cost Hamiltonian and every fixed $\gamma$,
\[
\sup_{\beta}
P_\Pi(\beta,\gamma)
\le
(N+1)\overline P_\Pi(\gamma).
\]
The mixer average from  satisfies
\[
 \overline P_\Pi(\gamma)
 =2^{-N}\sum_{x\in\Pi}\sum_{k=0}^N|F_k^{(\gamma)}(x)|^2.
\]
\end{corollary}

For a permutation $\sigma\in S_n$, let $x_\sigma$ be the corresponding
permutation matrix and define its shell-coherence factor by
\[
\mathcal C_\sigma(\gamma)
:=
\sum_{k=0}^{N}
|F_k^{(\gamma)}(x_\sigma)|^2.
\]
Then
\[
\overline P_\Pi(\gamma)
=
2^{-N}
\sum_{\sigma\in S_n}
\mathcal C_\sigma(\gamma).
\]

\subsubsection{Generating-function representation}

Introduce the generating function
\[
G_{\sigma,\gamma}(z)
=
\sum_{k=0}^{N}
F_k^{(\gamma)}(x_\sigma)z^k.
\]
Applying \cref{eq:ft} and the Krawtchouk generating identity
\eqref{eq:Kraw-gen} in dimension $N$ gives
\begin{equation}
 G_{\sigma,\gamma}(z)
 =2^{-N}\sum_y\phi_\gamma(y)
 (1+z)^{N-d(x_\sigma,y)}(1-z)^{d(x_\sigma,y)}.
 \label{eq:G-generating}
\end{equation}
Consequently,
\[
F_k^{(\gamma)}(x_\sigma)
=
[z^k]G_{\sigma,\gamma}(z),
\]
and Parseval on the unit circle gives
\begin{equation}
\mathcal C_\sigma(\gamma)
=
\frac{1}{2\pi}
\int_0^{2\pi}
|G_{\sigma,\gamma}(e^{i\theta})|^2\,d\theta.
\label{eq:C-parseval}
\end{equation}

For $\theta\in[0,2\pi]$, define
\[
a_\theta
=
\left|
\cos\frac{\theta}{2}
\right|,
\qquad
b_\theta
=
\left|
\sin\frac{\theta}{2}
\right|,
\qquad
L_\theta=a_\theta+b_\theta,
\]
and
\[
q_\theta
=
\frac{b_\theta}{a_\theta+b_\theta}.
\]
The factors in \cref{eq:G-generating} can then be absorbed into a
product Bernoulli measure. More precisely, there exists a phase
$\eta_{\sigma,\theta}(y)$ with
$|\eta_{\sigma,\theta}(y)|=1$ such that
\begin{equation}
G_{\sigma,\gamma}(e^{i\theta})
=
e^{i\vartheta_{\sigma,\theta}}
L_\theta^N
\mathbb E_{\nu_{\sigma,q_\theta}}
\left[
\eta_{\sigma,\theta}(x)
e^{-i\gamma C(x)}
\right].
\label{eq:G-probabilistic}
\end{equation}
Here, identify binary matrices with vectors in $\{0,1\}^{N}$.
\[
x=x_\sigma\oplus \zeta,
\qquad
\zeta=(\zeta_u)_{u=1}^N,
\qquad
\zeta_u\stackrel{\mathrm{ind}}{\sim}\mathrm{Bernoulli}(q_\theta).
\]
Accordingly, the entries of $x$ remain independent and each coordinate
has Bernoulli parameter either $q_\theta$ or $1-q_\theta$.

The elementary identity
\[
a_\theta^2+b_\theta^2=1
\]
also gives
\begin{equation}
L_\theta^2
=
\frac{1}{
q_\theta^2+(1-q_\theta)^2
}.
\label{eq:L-q-relation}
\end{equation}

\subsubsection{Random TSP phase ensemble}

We use the penalized cost in \cref{eq:tsp-total-cost}, the
transition matrix in \cref{eq:tsp-transition-matrix}, and the
ensemble and grid from \cref{sec:formulation-tsp-ensemble}.
Let $\mathcal F(x)=(F_e(x))_{e\in\mathcal E}$ collect the coefficients
of the independent random edge costs. For ATSP,
$\mathcal E=\{(i,j)\mid i\ne j\}$ and $F_{(i,j)}(x)=M_{ij}(x)$.
For symmetric TSP, $\mathcal E=\{\{i,j\}\mid i<j\}$ and
$F_{\{i,j\}}(x)=M_{ij}(x)+M_{ji}(x)$.
These features follow the off-diagonal edge variables in the stated ensemble.

For nonzero $a\in\mathbb Z_Q$, the phase-orthogonality identity
\eqref{eq:formulation-phase-orthogonality} yields
\begin{equation}
 \mathbb E_D\exp\left[-i\gamma_a\sum_{e\in\mathcal E}
 D_e\bigl(F_e(x)-F_e(x')\bigr)\right]
 =\mathbf1_{\{\mathcal F(x)=\mathcal F(x')\}}.
 \label{eq:random-phase-orthogonality}
\end{equation}
Feature differences have magnitude at most $2n<Q$, so equality modulo
$Q$ is equality over the integers. The deterministic penalty contributes
a unit-modulus phase to the terms surviving this average. Squaring
\cref{eq:G-probabilistic} and taking absolute values after averaging
therefore gives
\begin{equation}
 \mathbb E_D|G_{\sigma,\gamma_a}(e^{i\theta})|^2
 \le L_\theta^{2N}\Pr[\mathcal F(x)=\mathcal F(x')],
 \label{eq:G-collision-reduction}
\end{equation}
where $x$ and $x'$ are independent with law $\nu_{\sigma,q_\theta}$.

\subsubsection{Anti-concentration and rank estimates}

We first record two elementary probabilistic lemmas.

\begin{lemma}[Linear Bernoulli anti-concentration]
\label[lemma]{lem:linear-bernoulli-anticoncentration}
Let
\[
\xi=(\xi_1,\ldots,\xi_s)
\]
have independent Bernoulli coordinates, not necessarily identically
distributed, and suppose
\[
\max\{
\Pr[\xi_j=0],\Pr[\xi_j=1]
\}
\le\alpha<1
\]
for every $j$. If $A$ is a deterministic real matrix of rank $r$, then
\[
\sup_v
\Pr[A\xi=v]
\le
\alpha^r.
\]
\end{lemma}

\begin{proof}
Choose $r$ linearly independent columns of $A$. Conditional on all
remaining coordinates of $\xi$, the equation
\[
A\xi=v
\]
determines the $r$ selected coordinates uniquely, if it is solvable at
all. By independence, the probability of any prescribed assignment of
these coordinates is at most $\alpha^r$.
\end{proof}

\begin{lemma}[Rank tail for an inhomogeneous Bernoulli matrix]
\label[lemma]{lem:bernoulli-rank-tail}
Let $A$ be an $m\times n$ random matrix whose entries are independent
Bernoulli random variables and satisfy
\[
\max\{
\Pr[A_{ij}=0],\Pr[A_{ij}=1]
\}
\le\alpha<1.
\]
With rank taken over $\mathbb R$, for every $0\le r<\min(m,n+1)$,
\[
\Pr[
\operatorname{rank}(A)\le r
]
\le
{m\choose r}
\alpha^{(m-r)(n-r)}.
\]
\end{lemma}

\begin{proof}
If $\operatorname{rank}(A)\le r$, there exists a set of $r$ rows whose
span contains every row of $A$. Fix such a candidate set. Conditional
on these $r$ rows, their span is a subspace of dimension at most $r$ and
therefore codimension at least $n-r$.

Apply \Cref{lem:linear-bernoulli-anticoncentration} to a linear
map whose kernel is that span. Each remaining row lies in the span
with probability at most $\alpha^{n-r}$.
The probability that all remaining rows lie in the span is therefore at
most
\[
\alpha^{(m-r)(n-r)}.
\]
Finally, there are at most ${m\choose r}$ choices of the spanning rows.
\end{proof}

We can now establish the required transition-collision estimate.

\begin{lemma}[Quadratic transition-collision suppression]
\label[lemma]{lem:transition-collision}
Fix
\[
q_0\in(0,\tfrac12)
\]
and let $x,x'$ be independent random binary $n\times n$ matrices whose
entries are mutually independent and whose Bernoulli parameters belong
to
\[
[q_0,1-q_0].
\]
For either the ATSP or symmetric TSP feature map $\mathcal F$ above,
there is a constant $c_{\mathrm{col}}(q_0)>0$ such that
\[
 \Pr[\mathcal F(x)=\mathcal F(x')]
 \le e^{-c_{\mathrm{col}}n^2}
\]
for all sufficiently large $n$.
\end{lemma}

\begin{proof}
Set $\alpha=1-q_0$ and partition the city indices as $[n]=I\sqcup J$,
with $m=|I|=\lfloor n/2\rfloor$ and $\ell=|J|=n-m$.
Write $T_{\mathrm{tr}}=\mathsf S$ for ATSP and
$T_{\mathrm{tr}}=\mathsf S+\mathsf S^{\mathsf T}$ for symmetric TSP, where $\mathsf S$ is the
cyclic shift used in \cref{eq:tsp-transition-matrix}.
Its real nullity $s$ is zero for ATSP and at most two for symmetric TSP.
The latter follows from the recurrence $v_{j+1}=-v_{j-1}$ for a vector
in the kernel of $\mathsf S+\mathsf S^{\mathsf T}$.

Feature equality implies equality of the cross blocks
\[
 x_I T_{\mathrm{tr}}x_J^{\mathsf T}
 =x'_I T_{\mathrm{tr}}(x'_J)^{\mathsf T}.
\]
Let $A=x_I T_{\mathrm{tr}}$ and $A'=x'_I T_{\mathrm{tr}}$.
Then $\operatorname{rank}(A)\ge\operatorname{rank}(x_I)-s$.
Conditional on $x_I,x'_I$, each $j\in J$ must satisfy
\[
 [A,-A']\begin{pmatrix}x_j^{\mathsf T}\\(x'_j)^{\mathsf T}\end{pmatrix}=0.
\]
\Cref{lem:linear-bernoulli-anticoncentration} bounds this probability
by $\alpha^{\operatorname{rank}(A)}$. Independence of the row pairs gives
\[
 \Pr[\mathcal F(x)=\mathcal F(x')\mid x_I,x'_I]
 \le\alpha^{\ell\operatorname{rank}(A)}.
\]
Set $r_*=\lfloor m/2\rfloor$, taking $n$ large enough that $r_*>s$.
Splitting according to the rank of $x_I$ and applying
\Cref{lem:bernoulli-rank-tail} gives
\begin{equation}
 \Pr[\mathcal F(x)=\mathcal F(x')]
 \le\alpha^{\ell(r_*-s)}
 +\binom m{r_*-1}\alpha^{(m-r_*+1)(n-r_*+1)}.
 \label{eq:collision-rank-split}
\end{equation}
The exponents are $n^2/8+O(n)$ and $3n^2/16+O(n)$, respectively,
and $\binom m{r_*-1}\le2^m=e^{O(n)}$. Since $\alpha<1$, their sum
is at most $e^{-c_{\mathrm{col}}n^2}$ for a positive constant and all
sufficiently large $n$.
\end{proof}

\subsubsection{Uniform shell-coherence suppression}
\label{app:shell-suppression}

We next combine the collision estimate with the deterministic envelope
coming from the Hamming generating function.

Fix $q_0\in(0,\tfrac12)$.

If
\[
q_\theta\in[q_0,1-q_0],
\]
\Cref{lem:transition-collision} and
\cref{eq:G-collision-reduction} imply
\[
\mathbb E_D
|G_{\sigma,\gamma_a}(e^{i\theta})|^2
\le
L_\theta^{2N}
e^{-c_{\mathrm{col}}N}.
\]
Since
\[
L_\theta^2\le2,
\]
we obtain
\begin{equation}
\mathbb E_D
|G_{\sigma,\gamma_a}(e^{i\theta})|^2
\le
2^N e^{-c_{\mathrm{col}}N}.
\label{eq:G-central}
\end{equation}

For the complementary range of $q_\theta$, use
\cref{eq:L-q-relation}, if
\[
\min\{q_\theta,1-q_\theta\}<q_0,
\]
then
\[
L_\theta^2
\le
\frac{1}{
q_0^2+(1-q_0)^2
}
=
2\rho_0,
\]
where
\[
\rho_0
:=
\frac{1}{
2[q_0^2+(1-q_0)^2]
}
<1.
\]
The trivial bound on the expectation in
\cref{eq:G-probabilistic} therefore gives
\begin{equation}
|G_{\sigma,\gamma_a}(e^{i\theta})|^2
\le
2^N\rho_0^N
=
2^Ne^{-c_{\mathrm{mix}}N},
\label{eq:G-edge}
\end{equation}
with
\[
c_{\mathrm{mix}}
:=
-\log\rho_0
>0.
\]

Define
\[
c_*
:=
\min\{
c_{\mathrm{col}},
c_{\mathrm{mix}}
\}
>0.
\]
Equations~\eqref{eq:G-central} and
\eqref{eq:G-edge} imply, uniformly over
$\theta\in[0,2\pi]$,
\[
\mathbb E_D
|G_{\sigma,\gamma_a}(e^{i\theta})|^2
\le
2^Ne^{-c_*N}
\]
for every nonzero $\gamma_a\in\Gamma_Q$.

Integrating and using
\cref{eq:C-parseval} gives
\begin{equation}
\boxed{
\mathbb E_D
\mathcal C_\sigma(\gamma_a)
\le
2^Ne^{-c_*N}.
}
\label{eq:expected-shell-coherence}
\end{equation}

Summing over the $n!$ feasible permutations,
\begin{align}
\mathbb E_D
\overline P_\Pi(\gamma_a)
&=
2^{-N}
\sum_{\sigma\in S_n}
\mathbb E_D
\mathcal C_\sigma(\gamma_a)
 \le
n!e^{-c_*N}.
\label{eq:expected-dephased-feasible}
\end{align}

\Cref{eq:expected-dephased-feasible} supplies the ensemble
input in \cref{eq:training-dephased-input} for both ATSP and
symmetric TSP. \Cref{cor:full-feasible-alignment} supplies
\cref{eq:training-alignment-input}. The Markov and union-bound
argument then establishes the simultaneous grid guarantee, including the exponential grid
cardinalities and  its adaptive sampling
consequence.

%app:lenses
\section{Alternative Views of Different Flavours}\label{app:lenses}
As summarised in Section \ref{sec:alternate}, our analysis develops four complementary approaches to
feasibility concentration. These are harmonic analysis,
parameter averaging, fourth-moment estimates, and light-cone
analysis. They examine how objective-dependent phases and
mixer dynamics concentrate probability on the feasible sector,
each resolving a different aspect of the same process. Here we document the complimentary viewpoints in detail.
\subsection{Route II. Parameter averaging and typical-angle bounds}
\label{app:avg-angles}

Cost-angle averaging organizes the depth-one output probabilities by
energy fibres. Configurations with equal cost contribute coherently
within each fibre. This gives a spectral expression for the mean feasible
probability and a typical-angle bound through Markov's inequality.
Throughout this subsection, the instance and mixer angle $\beta$ are fixed.

\begin{lemma}[Cost-angle average for a lattice spectrum]
\label[lemma]{lem:avg-gamma-lattice}
Assume $C(y)=\omega h(y)$ for $\omega>0$ and integers $h(y)$.
Integer and rational costs admit this representation after a global
rescaling, as used in discussions of QAOA parameter
periodicity~\cite{montanezbarrera2024universalqaoa}.
For each attained energy $E$, define
\[
  \mathcal L_E=\{y\in\{0,1\}^N\mid C(y)=E\},
  \qquad g_C=\max_E|\mathcal L_E|,
\]
and let
\[
  U_{xy}(\beta)=\prod_{j=1}^N
  \langle x_j|e^{-i\beta X}|y_j\rangle.
\]
For $\gamma\sim\mathrm{Unif}[0,2\pi/\omega]$,
\[
  \mu_x(\beta)
  =\mathbb E_\gamma|a_x(\beta,\gamma)|^2
  =2^{-N}\sum_E
       \left|\sum_{y\in\mathcal L_E}U_{xy}(\beta)\right|^2.
\]
Consequently,
\[
  \mu_\Pi(\beta)
  =\mathbb E_\gamma P_\Pi(\beta,\gamma)
  =\sum_{x\in\Pi}\mu_x(\beta)
  \le \min\{1,g_C|\Pi|2^{-N}\}.
\]
For a nondegenerate spectrum,
$\mu_x(\beta)=2^{-N}$ and $\mu_\Pi(\beta)=|\Pi|2^{-N}$.
\end{lemma}

\begin{proof}
The depth-one amplitude is
\[
  a_x(\beta,\gamma)
  =2^{-N/2}\sum_y U_{xy}(\beta)e^{-i\gamma\omega h(y)}.
\]
Thus
\begin{equation}
\label{eq:angle-avg}
  \mathbb E_\gamma|a_x|^2
  =2^{-N}\sum_{y,z}U_{xy}(\beta)\overline{U_{xz}(\beta)}
       \mathbb E_\gamma e^{-i\gamma\omega(h(y)-h(z))}.
\end{equation}
The phase average equals $\mathbf 1[h(y)=h(z)]$.
Grouping the surviving terms by their common energy gives the expression
for $\mu_x$. Cauchy--Schwarz within each fibre gives
\[
  \mu_x(\beta)
  \le 2^{-N}\sum_E|\mathcal L_E|
                 \sum_{y\in\mathcal L_E}|U_{xy}(\beta)|^2
  \le g_C2^{-N},
\]
using $\sum_y|U_{xy}(\beta)|^2=1$.
Summing over feasible outputs proves the bound on $\mu_\Pi$.
For singleton fibres, unitarity gives the stated uniform mean.
\end{proof}

\begin{remark}[Discrete grid averaging]
\label[remark]{rem:discrete-grid}
For integer $h$, let $L\ge1$ and
$\gamma_a=2\pi a/(\omega L)$, $a=0,\ldots,L-1$.
Discrete Fourier orthogonality gives
\[
  \frac1L\sum_{a=0}^{L-1}|a_x(\beta,\gamma_a)|^2
  =2^{-N}\sum_{r\in\mathbb Z_L}
      \left|\sum_{h(y)\equiv r\pmod L}U_{xy}(\beta)\right|^2.
\]
If $L>\max_y h(y)-\min_y h(y)$, congruence of attained energies
implies equality, and this average equals $\mu_x(\beta)$.
For $C(y)/\omega\in M^{-1}\mathbb Z$, apply the same construction
with lattice spacing $\omega/M$ and integer energies $MC(y)/\omega$.
The grid cardinality is chosen for the range of these integer energies.
\end{remark}

\begin{proposition}[Typical-angle upper bound]
\label[proposition]{prop:markov}
Under \Cref{lem:avg-gamma-lattice}, fix $\beta$ and $t>1$.
Then
\[
  \Pr_\gamma\bigl[P_\Pi(\beta,\gamma)
                   >t\mu_\Pi(\beta)\bigr]\le\frac1t.
\]
In particular, for $n\ge2$ and an integer $k\ge1$,
\[
  \Pr_\gamma\left[
    P_\Pi(\beta,\gamma)
    \le\min\{1,n^k g_C|\Pi|2^{-N}\}\right]
  \ge1-n^{-k}.
\]
For a nondegenerate spectrum the factor $g_C$ equals one.
\end{proposition}

\begin{proof}
Apply Markov's inequality to the nonnegative feasible probability
with mean $\mu_\Pi(\beta)$. If this mean is zero, the probability
vanishes almost surely. The second statement follows from the bound
on $\mu_\Pi(\beta)$ and $P_\Pi\le1$.
\end{proof}

\begin{lemma}[Riemann--Lebesgue averaging for real spectra]
\label[lemma]{lem:RL-averaging}
Let $C(y)\in\mathbb R$ and fix $\beta$. For a probability density
$f\in L^1(\mathbb R)$, define
\[
  f_s(\gamma)=s^{-1}f(\gamma/s),\qquad
  \widehat f(\xi)=\int_{\mathbb R}f(t)e^{-i\xi t}\,dt.
\]
For $\gamma\sim f_s$,
\[
  \mathbb E_\gamma|a_x(\beta,\gamma)|^2
  =2^{-N}\sum_{y,z}U_{xy}(\beta)\overline{U_{xz}(\beta)}
          \widehat f\bigl(s[C(y)-C(z)]\bigr).
\]
With $\mu_x$ and $\mu_\Pi$ defined by the energy-fibre sums in
\Cref{lem:avg-gamma-lattice},
\[
  \lim_{s\to\infty}\mathbb E_\gamma|a_x(\beta,\gamma)|^2
  =\mu_x(\beta),\qquad
  \lim_{s\to\infty}\mathbb E_\gamma P_\Pi(\beta,\gamma)
  =\mu_\Pi(\beta).
\]
For a nondegenerate spectrum these limits are $2^{-N}$ and
$|\Pi|2^{-N}$, respectively.
\end{lemma}

\begin{proof}
Substitute $\gamma=st$ in the phase average.
The Riemann--Lebesgue lemma gives
$\widehat f(s\Delta)\to0$ for every fixed $\Delta\ne0$,
and $\widehat f(0)=1$. Taking the limit in the finite sum retains
precisely the pairs with equal energies.
\end{proof}

\begin{remark}[Uniform and Gaussian averaging]
\label[remark]{rem:RL-examples}
If the spectrum contains at least two distinct energies, define
\[
  \Delta_{\min}
  =\min_{C(y)\ne C(z)}|C(y)-C(z)|>0.
\]
For $\gamma\sim\mathrm{Unif}[-S,S]$ and $\Delta\ne0$,
\[
  \mathbb E e^{-i\gamma\Delta}
  =\frac{\sin(S\Delta)}{S\Delta},\qquad
  \left|\mathbb E e^{-i\gamma\Delta}\right|
  \le\min\{1,(S\Delta_{\min})^{-1}\}.
\]
For $\gamma\sim\mathcal N(0,\sigma^2)$,
\[
  \mathbb E e^{-i\gamma\Delta}=e^{-\sigma^2\Delta^2/2},\qquad
  \left|\mathbb E e^{-i\gamma\Delta}\right|
  \le e^{-\sigma^2\Delta_{\min}^2/2}
  \quad(\Delta\ne0).
\]
Writing $r$ for the corresponding upper bound gives the finite-window
estimate
\[
  \left|\mathbb E_\gamma P_\Pi(\beta,\gamma)-\mu_\Pi(\beta)\right|
  \le r\,2^{-N}\sum_{x\in\Pi}
       \sum_{C(y)\ne C(z)}|U_{xy}(\beta)U_{xz}(\beta)|.
\]
Equal-energy contributions remain in $\mu_\Pi$ at every window size.
For a constant spectrum the unequal-energy sum is empty and the mean
already equals $\mu_\Pi$.
\end{remark}

These estimates describe the distribution over the cost angle for a
fixed instance and mixer angle. The simultaneous parameter guarantees
in the main text use objective-ensemble averaging, mixer alignment,
and a union bound over the prescribed cost-angle family.

\subsection{Route III. Fourth moments and energy-sum resonances}
\label{app:L4}

The fourth moment of the output amplitudes controls the probability
assigned to a feasible sector. This norm-based approach belongs to the
Boolean harmonic-analysis framework discussed
in~\cite{ODonnell2014}. With the amplitude notation of
\cref{eq:prob_def}, Cauchy--Schwarz gives
\begin{equation}
  P_\Pi(\beta,\gamma)
  \le\sqrt{|\Pi|}\left(\sum_x|a_x(\beta,\gamma)|^4\right)^{1/2}.
  \label{eq:CS-L4}
\end{equation}
Cost-angle averaging expresses this fourth moment through resonances
between sums of two energies.

\begin{lemma}[Fourth moment under cost-angle averaging]
\label[lemma]{lem:L4-avg-gamma-correct}
Assume $C(y)=\omega h(y)$ with $h(y)\in\mathbb Z$ and $\omega>0$,
and let $\gamma\sim\mathrm{Unif}[0,2\pi/\omega]$.
Fix $\beta$ and define
\[
  \mathcal R_E=\{(y,z)\in(\{0,1\}^N)^2\mid C(y)+C(z)=E\},
  \qquad r_C=\max_E|\mathcal R_E|.
\]
The pairs are ordered. With $U_{xy}(\beta)$ as in
\Cref{lem:avg-gamma-lattice}, the averaged fourth moment satisfies
\[
  M_4(\beta)
  =\mathbb E_\gamma\sum_x|a_x(\beta,\gamma)|^4
  =2^{-2N}\sum_x\sum_E
      \left|\sum_{(y,z)\in\mathcal R_E}
                     U_{xy}(\beta)U_{xz}(\beta)\right|^2.
\]
In particular,
\[
  M_4(\beta)\le B_C(\beta),\qquad
  B_C(\beta)=\min\left\{
       r_C2^{-N},\,
       \left(\frac{1+|\sin(2\beta)|}{2}\right)^N
       \right\}.
\]
Consequently,
\[
  \mathbb E_\gamma P_\Pi(\beta,\gamma)
  \le\min\{1,\sqrt{|\Pi|M_4(\beta)}\}
  \le\min\{1,\sqrt{|\Pi|B_C(\beta)}\}.
\]
\end{lemma}

\begin{proof}
Write
\[
  a_x(\beta,\gamma)
  =2^{-N/2}\sum_yU_{xy}(\beta)e^{-i\gamma C(y)},\qquad
  U_{xy}(\beta)=\prod_{j=1}^N
       \langle x_j|e^{-i\beta X}|y_j\rangle.
\]
The squared amplitude is
\begin{equation}
  |a_x(\beta,\gamma)|^2
  =2^{-N}\sum_{y_1,y_2}
      U_{xy_1}\overline{U_{xy_2}}
      e^{-i\gamma[C(y_1)-C(y_2)]}.
  \label{eq:ax-sq}
\end{equation}
Set
\[
  \Delta(\boldsymbol y)
  =C(y_1)-C(y_2)+C(y_3)-C(y_4),
  \qquad \boldsymbol y=(y_1,y_2,y_3,y_4).
\]
Squaring gives
\begin{equation}
  |a_x(\beta,\gamma)|^4
  =2^{-2N}\sum_{\boldsymbol y}
       U_{xy_1}\overline{U_{xy_2}}
       U_{xy_3}\overline{U_{xy_4}}
       e^{-i\gamma\Delta(\boldsymbol y)}.
  \label{eq:ax-fourth-expanded}
\end{equation}
Summing over all outputs yields
\begin{equation}
  \sum_x|a_x(\beta,\gamma)|^4
  =2^{-2N}\sum_x\sum_{\boldsymbol y}
       U_{xy_1}\overline{U_{xy_2}}
       U_{xy_3}\overline{U_{xy_4}}
       e^{-i\gamma\Delta(\boldsymbol y)}.
  \label{eq:fourth-moment-full}
\end{equation}
The lattice phase average equals
$\mathbf1[\Delta(\boldsymbol y)=0]$, so
\begin{equation}
  M_4(\beta)
  =2^{-2N}\sum_x
    \sum_{\substack{y_1,y_2,y_3,y_4\\
        C(y_1)+C(y_3)=C(y_2)+C(y_4)}}
       U_{xy_1}\overline{U_{xy_2}}
       U_{xy_3}\overline{U_{xy_4}}.
  \label{eq:gamma-avg-quadruples}
\end{equation}
Grouping the pairs $(y_1,y_3)$ and $(y_2,y_4)$ by their common
energy sum gives the formula for $M_4$. Cauchy--Schwarz within
each pair-energy fibre gives
\begin{align*}
  M_4(\beta)
  &\le 2^{-2N}\sum_x\sum_E|\mathcal R_E|
           \sum_{(y,z)\in\mathcal R_E}|U_{xy}|^2|U_{xz}|^2\\
  &\le r_C2^{-2N}\sum_x
                 \left(\sum_y|U_{xy}|^2\right)^2
   =r_C2^{-N}.
\end{align*}

A bound from the absolute mixer entries follows by applying the
triangle inequality to \cref{eq:gamma-avg-quadruples} and
summing over all quadruples,
\begin{equation}
  M_4(\beta)
  \le2^{-2N}\sum_x\sum_{y_1,y_2,y_3,y_4}
          |U_{xy_1}U_{xy_2}U_{xy_3}U_{xy_4}|.
  \label{eq:drop-balance}
\end{equation}
Let $c=|\cos\beta|$ and $s=|\sin\beta|$. Each row of the
one-qubit mixer has absolute row sum $c+s$. Hence
\begin{equation}
  \sum_x\sum_{y_1,y_2,y_3,y_4}
          |U_{xy_1}U_{xy_2}U_{xy_3}U_{xy_4}|
  =2^N(c+s)^{4N},
  \label{eq:A-sum-factor}
\end{equation}
and
\begin{equation}
  M_4(\beta)\le\left(\frac{(c+s)^4}{2}\right)^N.
  \label{eq:raw-L4-bound}
\end{equation}
Normalization gives the sharper mixer estimate used in the statement.
The amplitude formula implies, for every cost angle,
\[
  |a_x(\beta,\gamma)|^2
  \le2^{-N}\left(\sum_y|U_{xy}(\beta)|\right)^2
  =\left(\frac{(c+s)^2}{2}\right)^N.
\]
Since $\sum_x|a_x|^2=1$,
\[
  \sum_x|a_x(\beta,\gamma)|^4
  \le\max_x|a_x(\beta,\gamma)|^2
  \le\left(\frac{1+|\sin(2\beta)|}{2}\right)^N.
\]
Averaging and combining the two bounds proves
$M_4(\beta)\le B_C(\beta)$. Finally, Jensen's inequality applied
to \cref{eq:CS-L4} gives the bound on the mean feasible probability.
\end{proof}

\begin{proposition}[Typical-angle fourth-moment bound]
\label[proposition]{prop:typical-L4}
Under \Cref{lem:L4-avg-gamma-correct}, fix $\beta$ and $t>1$.
Then
\[
  \Pr_\gamma\left[\sum_x|a_x(\beta,\gamma)|^4
                         \ge tM_4(\beta)\right]\le\frac1t.
\]
Consequently, for $n\ge2$ and an integer $k\ge1$,
\[
  \Pr_\gamma\left[
     P_\Pi(\beta,\gamma)
       \le\min\{1,n^{k/2}\sqrt{|\Pi|B_C(\beta)}\}\right]
  \ge1-n^{-k}.
\]
\end{proposition}

\begin{proof}
The fourth moment is nonnegative and has mean $M_4(\beta)>0$.
Apply Markov's inequality and then \cref{eq:CS-L4}, using
$M_4(\beta)\le B_C(\beta)$.
\end{proof}

The mixer estimate also gives a deterministic feasibility bound,
\[
  P_\Pi(\beta,\gamma)
  \le\min\left\{1,\sqrt{|\Pi|}
       \left(\frac{1+|\sin(2\beta)|}{2}\right)^{N/2}\right\}.
\]
For permutation constraints with $N=n^2$ and $|\Pi|=n!$, this is
$e^{-\Omega(n^2)}$ uniformly on any mixer-angle domain satisfying
$|\sin(2\beta)|\le1-\varepsilon$ for a fixed $\varepsilon>0$.
The energy-sum formula for $M_4$ retains the objective dependence
across the full mixer-angle range. Estimates of its resonance sums
provide additional inputs to the typical-angle bound. The main-text
bounds over all mixer angles follow from the objective-ensemble
estimates and mixer alignment.

\subsection{Route IV -- Local--Light--Cone Structure}
\label{app:lightcone}

The harmonic and probabilistic estimates developed above
control feasibility through Walsh-shell alignment and
collisions of TSP transition features. The light-cone
formulation connects these estimates to the propagation
and overlap of feasibility constraints. It also provides
bounds for partially overlapping constraint families. The support analysis follows the locality framework
of interacting quantum systems
\cite{LiebRobinson1972,Hastings2010Locality}.
For the commuting diagonal cost Hamiltonians considered
here, support growth is determined directly by the
interaction hypergraph.

We use the generic circuit $U_{p,D}$, initial state
$|\psi_0\rangle=|+\rangle^{\otimes N}$, and feasible
projector $\widehat\Pi$ defined in \cref{eq:formulation-feasible-probability}.
Write
\[
H_{C,D}=\sum_a H_{a,D}(Z),
\qquad
\widehat\Pi=\prod_{b=1}^{M}P_b,
\]
where the cost terms are commuting diagonal operators
and the $P_b$ are commuting constraint projectors.
The supports of the cost terms define the interaction
hypergraph.

\paragraph{Single-layer support growth.}

Let $O$ be supported on $S\subseteq[N]$.
The transverse-field mixer preserves its support.
During conjugation by the diagonal cost unitary,
every cost term disjoint from $S$ commutes through
the observable and cancels. Consequently,
\begin{equation}
\operatorname{supp}
\left(
U_{C,D}^{\dagger}U_X^{\dagger}
OU_XU_{C,D}
\right)
\subseteq N_1(S),
\label{eq:single-layer-support}
\end{equation}
where $N_1(S)$ contains $S$ and the supports of
all cost terms intersecting $S$.

For a constraint projector $P_b$, define
$Q_b=U_{1,D}^{\dagger}P_bU_{1,D}$. Then
\begin{equation}
\operatorname{supp}(Q_b)\subseteq S_b(1),
\qquad
S_b(1)=N_1\!\left(\operatorname{supp}P_b\right).
\label{eq:constraint-lightcone-p1}
\end{equation}

\subsubsection{Light-cone structure at arbitrary depth}
\label{sec:light-cone-arbitrary-depth}

Iteration of the single-layer support relation gives
\begin{equation}
\operatorname{supp}
\left(U_{p,D}^{\dagger}P_bU_{p,D}\right)
\subseteq S_b(p),
\qquad
S_b(p)=N_p\!\left(\operatorname{supp}P_b\right).
\label{eq:constraint-lightcone-p}
\end{equation}
Here $N_p$ denotes $p$ successive applications of
the interaction-neighborhood map. The sets $S_b(p)$
are structural support bounds determined by the
interaction hypergraph. Parameter-dependent
cancellations may reduce the actual supports.

At fixed depth, put
$Q_b=U_{p,D}^{\dagger}P_bU_{p,D}$.
Since unitary conjugation preserves products of
operators, the feasible probability satisfies
\begin{equation}
P_{\Pi,p}^{(\mathrm{gen})}
=
\langle\psi_0|
\prod_{b=1}^{M}Q_b
|\psi_0\rangle.
\label{eq:lightcone-feasible-probability}
\end{equation}

The overlap graph $\mathcal G_p$ has one vertex
for each constraint projector and an edge between
constraints whose structural backward cones
intersect. For the row-feasibility family, define
\begin{equation}
\chi_p
=
\max_{v\in[N]}
\left|
\{i\in[n]:v\in S_i(p)\}
\right|.
\label{eq:lightcone-overlap-number}
\end{equation}
The overlap number counts the maximum number of
evolved row conditions sharing an input coordinate.

In the standard permutation encoding, the squared
row and column penalties generate pairwise
interactions within every row and column.
Each complete row intersects every column.
It follows that
\[
N_1(\operatorname{Row}_i)=[N],
\]
and hence
\[
S_i(p)=[N],
\qquad
\chi_p=n,
\qquad p\ge1.
\]
The row and column constraint family contains
$2n$ projectors. The value $\chi_p=n$ refers
to the $n$ row conditions. For interaction hypergraphs with smaller backward
cones, $\chi_p$ records the corresponding
partial-overlap structure. Product-space inequalities
and block decompositions quantify the resulting
joint feasibility probabilities.

\subsubsection{Fubini--Tonelli and Finner--H\"older}

The product-space formulation allows constraint
events to share some input coordinates while
retaining control of their joint probability.
Fubini--Tonelli organizes the required integrations,
and Finner--H\"older accounts for the multiplicity
of overlapping dependencies.

\begin{lemma}[Fubini--Tonelli for product probability
spaces~\cite{Folland1999}]
\label[lemma]{lem:fubini-product}
Let $V$ be finite and let
\[
(\Omega_{\mathrm{prod}},\mu)
=
\bigotimes_{v\in V}
(\Omega_v,\mu_v).
\]
For every nonnegative measurable function $f$ and
every ordering $V=\{v_1,\ldots,v_{|V|}\}$,
\begin{equation}
\int_{\Omega_{\mathrm{prod}}}f\,d\mu
=
\int_{\Omega_{v_1}}
\cdots
\int_{\Omega_{v_{|V|}}}
f(\omega_{v_1},\ldots,\omega_{v_{|V|}})
\,d\mu_{v_{|V|}}\cdots d\mu_{v_1}.
\label{eq:fubini-iterated}
\end{equation}
The identity also holds for real or complex
$f\in L^1(\mu)$. In particular, for
$V=S\sqcup T$,
\[
\int_{\Omega_{\mathrm{prod}}}f\,d\mu
=
\int_{\Omega_S}
\left(
\int_{\Omega_T}
f(\omega_S,\omega_T)\,d\mu_T
\right)d\mu_S.
\]
\end{lemma}

\begin{proof}
The assertions are the finite-product
specializations of Tonelli's and Fubini's
theorems
\cite[\Cref{sec:formulation-tsp-ensemble}, Theorems~2.37 and~2.39]{Folland1999}.
\end{proof}

\begin{lemma}[Dependency form of
Finner--H\"older~\cite{Finner1992}]
\label[lemma]{lem:finner-light-cone}
Let $A_j\ge0$ be measurable functions on the
product probability space of
\Cref{lem:fubini-product}.
Suppose that $A_j$ depends only on coordinates
in $S_j\subseteq V$, and define
\[
\chi
=
\max_{v\in V}
|\{j:v\in S_j\}|.
\]
For $\chi\ge1$,
\begin{equation}
\mathbb E_\mu
\left[\prod_j A_j\right]
\le
\prod_j
\left(
\mathbb E_\mu[A_j^\chi]
\right)^{1/\chi}.
\label{eq:finner}
\end{equation}
For indicator functions $A_j=\mathbf1_{E_j}$,
\begin{equation}
\Pr_\mu\left[\bigcap_jE_j\right]
\le
\prod_j\Pr_\mu[E_j]^{1/\chi}.
\label{eq:finner-indicators}
\end{equation}
\end{lemma}

\begin{proof}
It suffices initially to consider bounded
nonnegative functions. We induct on the number
of independent coordinates, keeping the exponent
$\chi$ fixed.

Choose $v\in V$ and let
\[
J_v=\{j:v\in S_j\}.
\]
Integrate over $\omega_v$ first, using
\Cref{lem:fubini-product}. Since
$|J_v|\le\chi$, H\"older's inequality gives
\[
\int_{\Omega_v}
\prod_{j\in J_v}A_j\,d\mu_v
\le
\prod_{j\in J_v}
\left(
\int_{\Omega_v}A_j^\chi\,d\mu_v
\right)^{1/\chi},
\]
where constant factors complete the H\"older
exponents when $|J_v|<\chi$.

Define
\[
B_j=
\begin{cases}
\displaystyle
\left(
\int_{\Omega_v}A_j^\chi\,d\mu_v
\right)^{1/\chi},
&j\in J_v,\\[1ex]
A_j,&j\notin J_v.
\end{cases}
\]
The functions $B_j$ depend on the remaining
coordinates, whose overlap multiplicity is
at most $\chi$. The induction hypothesis yields
\[
\int\prod_jB_j\,d\mu_{V\setminus\{v\}}
\le
\prod_j
\left(
\int B_j^\chi\,
d\mu_{V\setminus\{v\}}
\right)^{1/\chi}.
\]
Fubini--Tonelli identifies the integrals on
the right with $\mathbb E_\mu[A_j^\chi]$.

Truncation and monotone convergence extend
the inequality to arbitrary nonnegative
measurable functions. For indicators,
$A_j^\chi=A_j$, proving the probability bound.
\end{proof}

The parameter $\chi$ quantifies the loss
from overlapping dependencies. Disjoint
coordinate families give $\chi=1$, recovering
independence. Partially overlapping families
retain the exponent $1/\chi$.

The product Bernoulli measures used in
\Cref{sec:probabilistic-p1-closure}
provide one application. After conditioning
on the rows in $I$, the cross-block collision
events indexed by $J$ depend on independent
row pairs. Their conditional dependency
multiplicity is one, as used in
\Cref{lem:transition-collision}. For partially overlapping collision systems,
\cref{eq:finner-indicators} controls the
joint event through the corresponding
dependency multiplicity. The required
anti-concentration estimates are supplied by
\Cref{lem:linear-bernoulli-anticoncentration}
and its appropriate specializations.

\subsubsection{Constraint blocks and quantum light cones}

We now organize the quantum feasibility conditions
into blocks of rows.

Let $\mathcal B$ be a partition of $[n]$.
For $B\in\mathcal B$, let $E_B$ be the event
that the rows in $B$ are one-hot and occupy
distinct columns. The projector $\Pi_B$
acts on their $n|B|$ qubits and has rank
\begin{equation}
\operatorname{rank}\Pi_B
=
(n)_{|B|}
=
\frac{n!}{(n-|B|)!}.
\label{eq:block-projector-rank}
\end{equation}
Extend each $\Pi_B$ by the identity on the
remaining rows. Full permutation feasibility
implies every block event, giving
\[
\widehat\Pi
\le
\prod_{B\in\mathcal B}\Pi_B.
\]

For an averaged circuit, consider a unitary
dilation with the environment registers included
in the backward light cones. Registers carrying
shared random variables contribute to every
backward cone that depends on them.

Let $q_{p,\mathcal B}$ be the number of colors
in a proper coloring of the block-cone overlap
graph. Blocks of the same color have disjoint
backward cones.

\begin{lemma}[Block light-cone probability bound]
\label[lemma]{lem:block-lightcone-bound}
Suppose the circuit dilation begins in a product
state on its system and environment inputs.
Let $\bar\rho_p$ be the averaged system output
and define
\[
\bar\rho_{p,B}
=
\operatorname{Tr}_{B^c}\bar\rho_p,
\qquad
\eta_{p,B}
=
\|\bar\rho_{p,B}\|_\infty.
\]
Then
\begin{equation}
\operatorname{Tr}(\widehat\Pi\bar\rho_p)
\le
\prod_{B\in\mathcal B}
\left[
\min\{1,(n)_{|B|}\eta_{p,B}\}
\right]^{1/q_{p,\mathcal B}}.
\label{eq:merged-block-lightcone-bound}
\end{equation}
\end{lemma}

\begin{proof}
Computational-basis measurement defines the
joint distribution of the block events $E_B$.

Let
$\mathcal B_1,\ldots,\mathcal B_q$
be the color classes, with
$q=q_{p,\mathcal B}$, and define
\[
A_c
=
\prod_{B\in\mathcal B_c}
\mathbf1_{E_B}.
\]
H\"older's inequality gives
\[
\Pr\left[
\bigcap_{B\in\mathcal B}E_B
\right]
=
\mathbb E\left[\prod_{c=1}^qA_c\right]
\le
\prod_{c=1}^q
(\mathbb E A_c)^{1/q}.
\]

Within each color class, the backward cones
are disjoint. Product initialization therefore
gives
\[
\mathbb E A_c
=
\prod_{B\in\mathcal B_c}\Pr[E_B].
\]
For every block,
\[
\begin{aligned}
\Pr[E_B]
&=
\operatorname{Tr}
(\Pi_B\bar\rho_{p,B}) \le
\min\left\{
1,\,
\operatorname{rank}(\Pi_B)
\|\bar\rho_{p,B}\|_\infty
\right\} =
\min\{1,(n)_{|B|}\eta_{p,B}\}.
\end{aligned}
\]
Combining the color-class estimates and
using full feasibility to bound the joint
block event proves the result.
\end{proof}

The coloring number $q_{p,\mathcal B}$ controls
the groups of disjoint quantum backward cones.
The multiplicity $\chi$ in
\Cref{lem:finner-light-cone}
controls dependencies in a classical product-space
representation. Together, these estimates
provide probability bounds for partially
overlapping constraint systems.

For the dense TSP encoding, all row cones
overlap after one layer. Taking the joint
block $B=[n]$ gives
\begin{equation}
\operatorname{Tr}(\widehat\Pi\bar\rho_p)
\le
n!\,\|\bar\rho_p\|_\infty.
\label{eq:joint-block-operator-bound}
\end{equation}
The operator norm in this bound is controlled
by the transition-feature fibres.

\subsubsection{Transition fibres and operator bounds}

The TSP transition matrix
$M(x)=x\mathsf Sx^{\mathsf T}$ is defined in
\cref{eq:tsp-transition-matrix}. As in
Appendix~\ref{sec:probabilistic-p1-closure}, write $\mathcal F(x)$ for the independent
objective features,
\[
F_{ij}(x)=M_{ij}(x),
\qquad i\ne j,
\]
for ATSP, and
\[
F_{\{i,j\}}(x)
=
M_{ij}(x)+M_{ji}(x),
\qquad i<j,
\]
for symmetric TSP.

Since each transition count lies between
zero and $n$, these features satisfy
\begin{equation}
|F_e(x)-F_e(x')|
\le2n.
\label{eq:feature-difference-bound}
\end{equation}

The ensemble of \Cref{sec:formulation-tsp-ensemble} has the
cost decomposition
\begin{equation}
C_D(x)
=
C_{\mathrm{pen}}(x)
+
\sum_eD_eF_e(x),
\qquad
D_e\stackrel{\mathrm{ind}}{\sim}
\operatorname{Unif}\{0,\ldots,Q-1\}.
\label{eq:lc-tsp-ensemble}
\end{equation}
Here $C_{\mathrm{pen}}$ contains the deterministic
row and column penalties, and the independent
variables $D_e$ are indexed by the appropriate
directed or undirected edges.

For each feature value $f$, define the
fibre projector
\[
\Pi_f
=
\sum_{x:\mathcal F(x)=f}|x\rangle\langle x|.
\]
The associated dephasing channel is
\begin{equation}
\mathcal D_F(\rho)
=
\sum_f\Pi_f\rho\Pi_f.
\label{eq:transition-dephasing-channel}
\end{equation}

The cross-block rank geometry used in
\Cref{lem:transition-collision}
also bounds the cardinalities of individual
feature fibres.

\begin{lemma}[Transition-fibre cardinality]
\label[lemma]{lem:lc-transition-fibre}
Set
\[
N=n^2,\qquad
m=\lfloor n/2\rfloor,\qquad
\ell=n-m,\qquad
d=m\ell.
\]
For ATSP let $s=0$. For symmetric TSP let
\[
s=
\operatorname{nullity}(\mathsf S+\mathsf S^{\mathsf T})
\le2.
\]
Then
\begin{equation}
\begin{aligned}
g_n
:=
\max_f\operatorname{rank}\Pi_f
&\le
2^{N-d+s\ell}
\sum_{r=0}^{m}
\binom mr2^{-(m-r)^2} \le
2^{N-d+m+s\ell}.
\end{aligned}
\label{eq:lc-transition-fibre-bound}
\end{equation}
\end{lemma}

\begin{proof}
Partition the city rows into $I$ and $J$,
with $|I|=m$ and $|J|=\ell$.
Write the row submatrices as $x_I$ and $x_J$.

For ATSP, fixing the complete feature vector
also fixes the cross-block transition matrix
\[
x_I\mathsf Sx_J^{\mathsf T}.
\]
We bound the number of binary pairs $(x_I,x_J)$
compatible with a fixed cross-block matrix.

Suppose $\operatorname{rank}(x_I)=r$.
There are at most
\[
\binom mr2^{nr+r(m-r)}
\]
such binary matrices $x_I$.

Indeed, choose $r$ spanning rows and their
$n$ entries. Every remaining row lies in
their $r$-dimensional real span.
An injective projection onto $r$ coordinates
shows that this span contains at most $2^r$
binary vectors.

For fixed $x_I$, the matrix $x_I\mathsf S$ has rank $r$.
Each row of $x_J$ satisfies a rank-$r$ linear
system determined by the prescribed
cross-block matrix. There are at most
$2^{n-r}$ binary solutions per row, giving
at most
\[
2^{(n-r)\ell}
\]
choices for $x_J$.

The combined exponent is
\[
\begin{aligned}
nr+r(m-r)+(n-r)\ell
&=
N-d-(m-r)^2.
\end{aligned}
\]
Summing over the possible ranks gives
\[
g_n
\le
2^{N-d}
\sum_{r=0}^{m}
\binom mr2^{-(m-r)^2}
\]
for ATSP.

For symmetric TSP, the cross-block feature
matrix is
\[
x_I(\mathsf S+\mathsf S^{\mathsf T})x_J^{\mathsf T}.
\]
Its left factor has rank at least $r-s$,
where
\[
s=
\operatorname{nullity}(\mathsf S+\mathsf S^{\mathsf T})
\le2.
\]
The corresponding count gains at most
$s\ell$ in the exponent. Hence
\[
g_n
\le
2^{N-d+s\ell}
\sum_{r=0}^{m}
\binom mr2^{-(m-r)^2}.
\]
Finally,
\[
\sum_{r=0}^{m}
\binom mr2^{-(m-r)^2}
\le
\sum_{r=0}^{m}\binom mr
=
2^m,
\]
which proves the stated bound.
\end{proof}

For the uniform initial density operator
\[
\rho_0
=
|+\rangle\langle+|^{\otimes N},
\]
the restriction of
$\mathcal D_F(\rho_0)$
to a feature fibre is rank one. Its
nonzero eigenvalue equals the fibre
cardinality divided by $2^N$.

\Cref{lem:lc-transition-fibre}
therefore gives
\begin{equation}
\mathcal D_F(\rho_0)
\preceq r_n I,
\qquad
r_n=2^{-d+m+s\ell}.
\label{eq:first-dephasing-operator-bound}
\end{equation}

This supplies an operator-norm bound for
the feature-dephased state and, through
\cref{eq:joint-block-operator-bound},
for its feasible probability.

\subsubsection{Linked cost schedules and layerwise dephasing}

We use the linked cost-angle schedules introduced
in the linear-depth analysis.
Fix a depth budget $P\ge1$, let
$b_n=4n+1$, and choose a prime $Q>b_n^P$.
Write (repeating Eq. \ref{eq:linked-cost-angle-schedule} for readability)
\begin{equation}
a_j(t)
=
t b_n^{j-1}\pmod Q,
\qquad
\gamma_j(t)
=
\frac{2\pi a_j(t)}{Q},
\qquad
t\in\mathbb Z_Q,\quad 1\le j\le P.
\label{eq:linked-cost-angle-schedule2}
\end{equation}
The first $p$ angles specify the schedule at
depth $p\le P$. Mixer angles remain
independently variable.

The base-$b_n$ structure separates transition
differences contributed by distinct cost layers.

\begin{lemma}[Layerwise separation of transition differences]
\label[lemma]{lem:linked-layer-separation}
For $t\ne0$, $p\le P$, and integers
$|z_j|\le2n$,
\begin{equation}
\sum_{j=1}^{p}a_j(t)z_j
\equiv0\pmod Q
\quad\Longleftrightarrow\quad
z_j=0
\quad\text{for every }j.
\label{eq:linked-no-cancellation}
\end{equation}
\end{lemma}

\begin{proof}
Since $Q$ is prime and $t\ne0$,
$t$ is invertible modulo $Q$.
The congruence is therefore equivalent to
\[
\sum_{j=1}^{p}b_n^{j-1}z_j
\equiv0\pmod Q.
\]

The absolute value of this integer is bounded by
\[
2n\sum_{j=1}^{p}b_n^{j-1}
=
\frac{b_n^p-1}{2}
<
Q.
\]
Thus the integer itself must vanish.

Suppose $z_k$ is the last nonzero coefficient.
Its contribution has magnitude at least
$b_n^{k-1}$. The sum of all preceding
contributions has magnitude at most
\[
2n\sum_{j=1}^{k-1}b_n^{j-1}
=
\frac{b_n^{k-1}-1}{2}
<
b_n^{k-1}.
\]
Cancellation is impossible. Hence every
coefficient vanishes.
\end{proof}

\begin{proposition}[Averaged evolution under linked schedules]
\label[proposition]{prop:linked-averaged-channel}
Fix $t\ne0$, $p\le P$, and mixer angles
independent of the sampled instance $D$.
Let $H_{\mathrm{pen}}$ be the diagonal operator
with entries $C_{\mathrm{pen}}(x)$, and set
\[
V_{j,t}
=
e^{-i\gamma_j(t)H_{\mathrm{pen}}},
\]
\[
\Phi_{j,t}(\rho)
=
U_X(\beta_j)
V_{j,t}
\mathcal D_F(\rho)
V_{j,t}^{\dagger}
U_X(\beta_j)^\dagger.
\]
Then
\begin{equation}
\begin{aligned}
\bar\rho_{p,t}
&:=
\mathbb E_D
\left[
U_{p,D}\rho_0U_{p,D}^{\dagger}
\right]\\
&=
\Phi_{p,t}\circ\cdots\circ
\Phi_{1,t}(\rho_0).
\end{aligned}
\label{eq:linked-channel-identity}
\end{equation}
Moreover,
\begin{equation}
\bar\rho_{p,t}\preceq r_nI,
\qquad
\mathbb E_D
P_{\Pi,p}^{(\mathrm{gen})}
(\boldsymbol\beta,\boldsymbol\gamma(t);D)
\le n!r_n.
\label{eq:linked-depth-probability-mean}
\end{equation}
\end{proposition}

\begin{proof}
Expand a density-matrix element of the
depth-$p$ circuit into paired computational-basis
histories
\[
x_0,\ldots,x_{p-1},
\qquad
x'_0,\ldots,x'_{p-1}.
\]
\Cref{lem:fubini-product} permits the
finite history sums and instance expectations
to be exchanged.

By \cref{eq:lc-tsp-ensemble}, the
objective-dependent part of the relative
phase is
\[
\exp\left[
-\frac{2\pi i}{Q}
\sum_eD_e
\sum_{j=1}^{p}a_j(t)
\bigl(
F_e(x_{j-1})-F_e(x'_{j-1})
\bigr)
\right].
\]

Independence and uniformity of the edge costs
give
\begin{align*}
&\mathbb E_D
\exp\left[
-\frac{2\pi i}{Q}
\sum_eD_e
\sum_{j=1}^{p}a_j(t)
\bigl(
F_e(x_{j-1})-F_e(x'_{j-1})
\bigr)
\right]
\\
&\quad=
\prod_e
\mathbf1_{\left\{
\sum_{j=1}^{p}a_j(t)
[F_e(x_{j-1})-F_e(x'_{j-1})]
\equiv0\pmod Q
\right\}}.
\end{align*}

The feature differences satisfy
\cref{eq:feature-difference-bound}.
\Cref{lem:linked-layer-separation}
therefore reduces the product to
\[
\prod_{j=1}^{p}
\mathbf1_{\{
\mathcal F(x_{j-1})=\mathcal F(x'_{j-1})
\}}.
\]
These indicators are exactly the restrictions
imposed by the feature-dephasing channel
in \cref{eq:transition-dephasing-channel} at each
cost layer.

The deterministic penalty phases and mixer
matrix elements remain in the history
expansion. Collecting the terms layer by
layer gives
\[
\bar\rho_{p,t}
=
\Phi_{p,t}\circ\cdots\circ
\Phi_{1,t}(\rho_0),
\]
establishing
\cref{eq:linked-channel-identity}. By \cref{eq:first-dephasing-operator-bound},
\[
\mathcal D_F(\rho_0)\preceq r_nI.
\]
Unitary conjugation preserves the operator
order, and $\mathcal D_F$ is positive and
unital. Every map $\Phi_{j,t}$
therefore satisfies
\[
\rho\preceq r_nI
\quad\Longrightarrow\quad
\Phi_{j,t}(\rho)\preceq r_nI.
\]
Iteration gives
\[
\bar\rho_{p,t}\preceq r_nI.
\]

Finally, the feasible projector has rank
$n!$. \Cref{eq:joint-block-operator-bound}
then yields
\[
\begin{aligned}
\mathbb E_D
P_{\Pi,p}^{(\mathrm{gen})}
=
\operatorname{Tr}(\widehat\Pi\bar\rho_{p,t})
\le
n!\|\bar\rho_{p,t}\|_\infty \le n!r_n.
\end{aligned}
\]
This proves
\cref{eq:linked-depth-probability-mean}.
\end{proof}

The linked-schedule identity gives the ensemble
average for circuits using the same sampled
objective in every cost layer. The layerwise
feature constraints follow from
\Cref{lem:linked-layer-separation}. Each feature dephasing admits a unitary
dilation with an environment register.
The feature operators commute, so their
dephasing channels compose to
$\mathcal D_F$. Independent environment
inputs supply a product-state dilation.
Their supports are included in the
corresponding backward cones, as required
by \Cref{lem:block-lightcone-bound}. For constraint families with partial overlap,
the Finner--H\"older and block light-cone
bounds express joint feasibility through
the overlap multiplicities and block
probabilities. For the dense TSP encoding, the row
light cones cover the full register after
one layer. The joint-block bound
\eqref{eq:joint-block-operator-bound}
and transition-fibre estimate
\eqref{eq:lc-transition-fibre-bound}
give the quantitative suppression
\eqref{eq:linked-depth-probability-mean}
throughout the stated depth range.

\pdfbookmark[1]{Data Availability Statement}{data-availability}
\section*{Data Availability Statement} All data in this paper and the Python implementation in Qiskit are made available here \url{https://doi.org/10.5281/zenodo.17701648}.

\pdfbookmark[1]{Conflict of Interests}{conflicts}
\section*{Conflict of Interests}
All authors declare no competing interests.

\pdfbookmark[1]{Author contributions}{author-contributions}
\section*{Author contributions}
C.O. wrote the main paper and prepared all the figures, K.M. provided guidance during the project. All the authors reviewed the paper.

\pdfbookmark[1]{Funding}{funding}
\section*{Funding}
This research received no external funding.

\printbibliography
\end{document}